\documentclass[11pt]{article}
\usepackage[numbers]{natbib}
\usepackage{packages}
\usepackage{editing-macros}
\usepackage{formatting}
\usepackage{statistics-macros}

\usepackage{booktabs}
\usepackage{mathtools}
\usepackage{mathrsfs}

\begin{document}

\abovedisplayskip=8pt plus0pt minus3pt
\belowdisplayskip=8pt plus0pt minus3pt

\begin{center}
  {\huge Empirical Likelihood for Nonsmooth Functionals} \\
  \vspace{.5cm} {\Large Hongseok Namkoong} \\
  \vspace{.2cm}
  {\large Decision, Risk, and Operations Division, Columbia Business School} \\
  \vspace{.2cm}
  \texttt{namkoong@gsb.columbia.edu}
\end{center}

\begin{abstract}%
  Empirical likelihood is an attractive inferential framework that respects
natural parameter boundaries, but existing approaches typically require
smoothness of the functional and miscalibrate substantially when these
assumptions are violated.  For the optimal-value functional central to policy
evaluation, smoothness holds only when the optimum is unique---a condition
that fails exactly when rigorous inference is most needed where more complex
policies have modest gains.  In this work, we develop a bootstrap empirical
likelihood method for partially nonsmooth functionals.  Our analytic workhorse
is a geometric reduction of the profile likelihood to the distance between the
score mean and a level set whose shape (a tangent cone given by
nonsmoothness patterns) determines the asymptotic distribution.  Unlike the
classical proof technology based on Taylor expansions on the dual optima, our
geometric approach leverages properties of a deterministic convex program and
can directly apply to nonsmooth functionals. Since the ordinary bootstrap is
not valid in the presence of nonsmoothness, we derive a corrected multiplier
bootstrap approach that adapts to the unknown level-set geometry.


\end{abstract}


\newcommand{\EL}{\mathrm{EL}}
\newcommand{\KL}{\mathrm{KL}}
\newcommand{\BL}{\mathrm{BL}}
\newcommand{\dd}{\,\mathrm{d}}


\section{Introduction}
\label{sec:intro}

Empirical likelihood (EL) is a nonparametric method of inference that retains
many of the virtues of parametric likelihood without committing to a
parametric model.  Given data $X_1, \dots, X_n$, one maximizes the product
$\prod_{i=1}^n p_i$ over probability vectors $(p_1, \dots, p_n)$ subject to
the constraint that a parameter of interest takes a hypothesized value.  The
resulting likelihood-ratio statistic converges to $\chi^2$ under the null,
yielding confidence regions that are automatically shaped by the data, respect
natural parameter boundaries, and satisfy a Bartlett
correction~\citep{Owen01}.  Because the construction requires no variance
estimation and no distributional specification, EL confidence regions have
complementary advantages to normal-theory counterparts in finite samples: they
are range-respecting, transformation-invariant, and self-normalizing.  These
properties have made EL an attractive foundation for inference on means,
quantiles, regression parameters, and smooth functions thereof.

For a decision-maker choosing among competing options under uncertainty, the
appeal of EL is especially direct.  The shape adaptation means that confidence
regions contract along directions the data resolve well and expand where they
do not, so the resulting bounds on decision-relevant quantities are neither
artificially wide nor miscalibrated.  This paper is motivated by whether these
advantages can be extended to a class of problems that is central to modern
data-driven decision-making but falls outside the scope of existing EL theory:
inference on the average reward under the best available policy, which often
guides investment decisions for interventions with an initial fixed setup
cost.

\paragraph{Motivating application} A large and growing class of decision
problems requires evaluating candidate policies from data collected under a
different, unknown behavioral policy.  In healthcare, a clinician considering
a new treatment assignment rule sees only the outcomes of the assignments that
were actually made; in e-commerce, a platform testing pricing strategies
observes revenue only under the prices that were historically
charged. Formally, the observed data consist of covariates, an action or treatment
drawn from the behavioral policy, and an outcome.  The decision-maker has
several candidate policies, each mapping covariates to actions, and the
quantity of interest is the value of the best available option,
\begin{equation}
  \tau(P_0)
  = \max_{1\le j\le J}
  \left\{ \theta_j(P_0) \defeq \E_{P_0}[Y^{\pi_j}]
    \right\},
  \label{eq:intro-target}
\end{equation}
where $Y^{\pi_j}$ is the potential outcome under policy~$\pi_j$.

Because the counterfactual outcomes are not directly observed, estimating each
policy value requires modeling nuisance functions---the propensity score and
the conditional outcome means---making the problem genuinely semiparametric.
Over the past decade, the double/debiased machine learning (DML) literature
has provided a clean reduction of this
difficulty~\citep{ChernozhukovChDeDuHaNeRo18}.  After splitting the sample and
fitting nuisance models on an auxiliary fold, one forms augmented
inverse-probability weighting (AIPW) scores on the evaluation fold:
\begin{equation}
  X_{ni,j}
  =
  \what m_{\pi_j(W_i)}(W_i)
  + \frac{\onevec\{A_i = \pi_j(W_i)\}}
  {\what e_{\pi_j(W_i)}(W_i)}
  \bigl(Y_i - \what m_{A_i}(W_i)\bigr),
  \qquad j = 1, \dots, J.
  \label{eq:intro-aipw}
\end{equation}
Under standard overlap and rate conditions on the nuisance
estimators~\citep{ChernozhukovChDeDuHaNeRo18, AtheyWa21, KallusUe20}, the
score vector satisfies a central limit theorem centered at the true
policy-value vector, and its sample mean converges at $\sqrt{n}$-rate.  The
debiasing step thus converts a semiparametric problem into a nonparametric
one: inference on the maximum of a $\sqrt{n}$-normal score mean.  Given these
score vectors, it is natural to apply empirical likelihood.

\paragraph{Why are ties relevant?}

In many applied settings, the primary question is not ``which policy is
best?''\ but ``does the best available policy improve enough over a
simple default to justify its implementation cost?''  A decision-maker
comparing a complex personalized treatment rule (which requires expensive
real-time data collection and infrastructure) against a simpler rule that
ignores certain covariates is performing exactly this kind of opportunity
sizing.  A growing body of empirical work suggests that the gains from
personalization are often modest: in education,
personalized tutoring policies frequently offer only marginal
improvements over static
alternatives~\citep{NieReBr23}; in
healthcare, heterogeneous treatment effects are routinely small relative
to average effects~\citep{ChernozhukovDeDuFe18}; and in
welfare-maximizing policy design, the improvement from targeting over a
uniform rule can be negligible~\citep{KitagawaTe18}.  When several
candidate policies achieve similar expected rewards, the set of
optimal policies contains multiple elements---and it is in this
near-tied regime that a rigorous confidence bound on the best attainable
value is most decision-relevant.  It is also where existing methods
break down.

\paragraph{Limitations of existing EL approaches}

The existing EL machinery, however, does not accommodate the nonsmooth max
functional~\eqref{eq:intro-target}.  Both classical EL results for nonlinear
functionals~\citep{Owen01, DuchiGlNa21} and semiparametric EL
extensions~\citep{HjortMcKe09, BravoEsVa20} require that the target map be
smooth.  \citet{MolanesLopezVaKe09} study EL for non-smooth criterion
functions, but their framework targets U-statistics and related criteria rather
than the directionally differentiable max functional that arises in policy
evaluation.  Applying them to~\eqref{eq:intro-target} forces one of two
workarounds.  The first is \emph{projected joint EL}: build a joint confidence
region for the full policy-value vector and project onto the maximum.  This is
valid but pays a dimension-$J$ critical value for a scalar target; at the
$95\%$ level with twenty policies, the resulting bound is significantly wider
than necessary (Proposition~\ref{prop:unique-inflation}).  The second is
\emph{selected-policy EL}: estimate the empirically best policy and run
one-dimensional EL on that coordinate.  This ignores selection uncertainty and
undercovers badly near ties.

Beyond the conservatism-versus-coverage tradeoff, semiparametric bootstrap
calibration can create a computational bottleneck.  The bootstrap method
of~\citet{HjortMcKe09} requires refitting all nuisance models inside every
resample, turning a moderate inferential problem into an expensive learning
problem repeated hundreds of times.  For smooth targets,
\citet{BravoEsVa20} sidestep this cost by replacing the original moment
with an influence-function-corrected equation, restoring the Wilks property
without resampling-level refitting.  Their fix, however, relies on
smoothness of the target functional and does not extend to the nonsmooth
max functional~\eqref{eq:intro-target}.

A further difficulty is nonregularity.  When one policy is uniquely
optimal, the max functional is Hadamard differentiable (HD) and a
standard $\chi^2$ calibration applies~\citep{DuchiGlNa21}.  But the
uniqueness assumption is itself at least as hard to verify as the
inferential task: certifying that one policy is strictly better than all
others requires the same kind of statistical evidence as the confidence
bound one is trying to construct.  When two or more policies are tied
or near-tied---the regime argued above to be the most
decision-relevant---the max functional is only Hadamard directionally
differentiable (HDD) and the geometry of the problem changes.  The local level set is no longer a
hyperplane but a \emph{cone} whose faces correspond to different
tie-breaking patterns among the optimal policies.  The standard
$\chi_1^2$ calibration, which is correct for hyperplanes,
systematically underestimates the critical value for cones because it
ignores this combinatorial face structure.  The ordinary bootstrap
inherits this failure: it always estimates the level set as a hyperplane
(since sample ties occur with probability zero), and therefore converges
to $\chi_1^2$ regardless of the true geometry~\citep{FangSa19}.  Our
experiments in Section~\ref{sec:sim} show that coverage drops well
below the nominal level when two or more policies are near-tied.  Corrected bootstrap methods for
nonsmooth functionals exist~\citep{FangSa19, HongLi18, HongLi20,
  BeutnerZa16}, but they target plug-in estimators rather than
profile-divergence statistics and do not inherit the shape-respecting,
self-normalizing properties that make EL attractive in the first place.

\subsection{Contributions}

We develop a \emph{score-profile empirical likelihood} that targets the
optimal value directly. When applied to our motivating problem, our method
operates on the
AIPW score vectors~\eqref{eq:intro-aipw} produced by a single nuisance fit and
requires no refitting inside the bootstrap loop.

Our central technical contribution is a \emph{primal} geometric proof strategy
that replaces the dual Lagrange-multiplier expansion used in classical EL
theory.  In the standard approach---developed by~\citet{Owen01} and extended
to semiparametric settings by~\citet{HjortMcKe09}
and~\citet{BravoEsVa20}---one writes the EL ratio as a function of a dual
parameter (the Lagrange multiplier enforcing the moment constraint), performs
a Taylor expansion of the dual optimality conditions around the true parameter
value, and inverts to obtain an asymptotic $\chi^2$ limit.  This dual route
works well for smooth estimating equations, but it encounters two obstacles
for nonsmooth targets like~\eqref{eq:intro-target}: the profile over the
nonsmooth level set does not reduce to a single moment constraint with a
smooth Lagrangian, so the standard KKT expansion has no smooth analogue.

Our primal proof avoids this obstacle by working directly in the
\emph{weight space} of the empirical likelihood.  The key reduction
(Theorem~\ref{thm:exact-level-set}) shows that the Euclidean profile
statistic---the minimum $\ell_2$ distance from uniform weights to the
set of reweightings consistent with the target value---equals the
Mahalanobis distance from the score mean to the rescaled level set.
This is a finite-dimensional convex projection whose behavior depends
entirely on the geometry of the target map, not on the mechanics of a
dual expansion.  Whether the level set is locally a hyperplane or a
cone then determines the asymptotic distribution and bootstrap validity.
In the \emph{smooth} (HD) case, the level set converges to a
hyperplane, the Mahalanobis projection is a scalar ratio, and the
profile statistic converges to $\chi_1^2$.  An ordinary score
bootstrap---applied to the frozen score vectors, with no nuisance
refit---is consistent (Theorem~\ref{thm:hd-boot}).

In the \emph{nonsmooth} (HDD) case, the level set converges to a tangent
cone, the limit becomes the squared Mahalanobis distance from a Gaussian
vector to this cone, and the ordinary bootstrap fails.  We develop a
corrected calibration procedure based on a multiplier bootstrap approach
(Theorem~\ref{thm:hdd-boot}) that estimates the tangent cone from the data
and computes the bootstrap distance to the estimated cone.  The same
conclusions extend to general smooth $f$-divergences under a local
quadratic-geometry condition, because the primal argument never invokes a
Fenchel conjugate and therefore imposes no smoothness conditions on the
divergence beyond local quadratic behavior near the uniform distribution.

We note two differences from existing semiparametric EL
frameworks. First, the entire primal analysis is deterministic once the score
array is given: the geometric reduction is a statement about the structure of
a convex program, not about the probabilistic behavior of a dual optimizer.
This makes the extension to the bootstrap immediate---the same convex
projection is applied to the bootstrap score vectors, and cone estimation
handles the nonsmooth case.  By contrast, the standard dual approach (e.g.,
\citet{HjortMcKe09}) requires a stochastic expansion of the Lagrange
multiplier that must be controlled uniformly over bootstrap resamples, an
analysis whose difficulty scales with the complexity of the moment condition
and that has no natural extension to nonsmooth targets.

Second, because the geometric reduction operates on the score array rather
than on the original data, the bootstrap inherits the score-level structure:
once the AIPW scores have been formed, every subsequent computation is $O(nJ)$
per resample, with no reference to the nuisance models.  Any semiparametric EL
method whose bootstrap is coupled to the nuisance-estimation step---including
the framework of \citet{HjortMcKe09}---must refit all nuisance models inside
every resample, turning each bootstrap draw into a full semiparametric learning
problem.  Even in our small-scale simulations, this difference leads to
multiple orders of magnitude improvement in computational effort.

Although we develop the theory in detail for the max functional---the
leading case in policy evaluation and one whose combinatorial structure
yields closed-form profile bounds---the core results (the primal
geometric reduction, the level-set characterization, and both bootstrap
procedures) apply to any functional satisfying
Assumption~\ref{ass:levelset}, including other Hadamard directionally
differentiable maps such as quantile and supremum-norm functionals.


\section{Related work}
\label{section:related-work}

Our work connects four strands of literature.

\paragraph{Classical empirical likelihood.}  Since \citet{Owen88}'s seminal work
that developed the theory of empirical likelihood for means, several authors
have extended the framework to smooth functions of means and regression
settings, and established the $\chi^2$ calibration, Bartlett correctability,
and range-respecting property (see~\cite{Owen01} and references therein).  Our
primal proof strategy can be viewed as revisiting Owen's original
likelihood-ratio construction from the weight-space perspective rather than
the classical dual Lagrange multiplier perspective~\cite{Owen01}, and showing
that this viewpoint extends naturally to directionally differentiable maps
where the dual route encounters fundamental obstructions.  In particular,
compared to \citet{DuchiGlNa21}'s recent generalized EL result for Hadamard
differentiable functionals, our approach provides calibrated confidence
intervals even for Hadamard differentiable functionals such as the infimum
functional when there is multiple optima.
\citet{MolanesLopezVaKe09} extend EL to non-smooth criterion functions such as
absolute deviations and U-statistics, establishing a Wilks-type theorem under a
local smoothing device.  Their non-smoothness is of a different kind from ours:
it lies in the criterion used to define the EL constraint, whereas in our
setting the moment conditions are smooth but the functional mapping scores to
the target parameter is only directionally differentiable.

\paragraph{Semiparametric empirical likelihood.}  \citet{HjortMcKe09} extend
empirical likelihood to plug-in nuisance parameters and slower-than-$\sqrt n$
rates, establishing both a Wilks phenomenon and bootstrap validity for
semiparametric estimating equations.  Their results are broad and powerful but
require that the target be defined as the solution to a smooth moment
condition; the max functional~\eqref{eq:intro-target} does not have this
structure.  \citet{BravoEsVa20} show that naive two-step plug-in EL can fail
to have a $\chi^2$ limit even for smooth targets---the nuisance estimation
bias contaminates the EL ratio---and restore the Wilks property by replacing
the original moment with an influence-function-corrected equation.  Their
correction is essential for any practical semiparametric EL implementation,
but like \citet{HjortMcKe09} it applies to smooth estimating equations.  Our
contribution departs from both by profiling the \emph{target functional
  itself}---rather than an estimating equation whose solution equals the
target---and by handling the nonsmooth case where the level-set geometry is a
cone rather than a hyperplane.

\paragraph{Debiased inference and policy evaluation.}
\citet{ChernozhukovChDeDuHaNeRo18} develop the double/debiased machine
learning framework, providing general conditions under which
cross-fitted orthogonal scores achieve $\sqrt n$-normality for
semiparametric targets even when nuisance functions are estimated at
slower-than-parametric rates.  This framework supplies the score vectors
that our method takes as input.  \citet{AtheyWa21} adapt these ideas to
policy learning and off-policy evaluation, deriving regret bounds and
inference for individual policy values under treatment-effect
heterogeneity.  In a complementary direction, \citet{LiBr26} study
optimal value inference in the semiparametric setting and derive
debiased semiparametric efficiency guarantees under uniqueness of
optima.  Our contribution is in what happens after the scores are formed:
constructing a profile-divergence statistic for the nonlinear scalar
target and calibrating it correctly in both the smooth and nonsmooth
regimes.

\paragraph{Inference for nonsmooth functionals.}  When two or more policies are
tied for optimality, the max functional~\eqref{eq:intro-target} is only
(Hadamard) directionally differentiable, and the ordinary bootstrap fails.
\citet{FangSa19} develop a general bootstrap correction for directionally
differentiable maps, establishing consistency of a plug-in multiplier
bootstrap under cone-approximation conditions.  Their framework is the closest
antecedent to our nonsmooth theory; the difference is that they calibrate a
plug-in estimator while we calibrate a profile-divergence statistic.  The two
approaches draw on the same local cone information but produce different
inferential objects: a corrected confidence interval versus a directly
invertible profile bound.  \citet{HongLi18} provide bootstrap consistency
results for directionally differentiable functions with applications to
conditional value-at-risk and related risk measures; \citet{HongLi20} extend
the analysis to M-estimators of nonsmooth parameters.  Our work studies the
same directional differentiability structure but in the empirical likelihood
framework, which yields the simplex formulation of
Proposition~\ref{prop:max-lower} and the computational advantages of
score-level resampling.

The remainder of the paper is organized as follows.
Section~\ref{sec:geometry} introduces the score-profile framework and
establishes the central geometric reduction that connects the profile
statistic to a Mahalanobis distance to a level set.  Section~\ref{sec:regular}
develops the smooth (Hadamard differentiable) case, derives closed-form
profile bounds for the policy problem, and shows that an ordinary score
bootstrap is consistent.  Section~\ref{sec:nonregular} handles the nonsmooth
case of tied policies, develops the corrected multiplier bootstrap, and
describes the full inferential procedure.  Section~\ref{sec:sim} reports
simulation evidence on coverage, tightness, and computation time.


\section{Score-profile framework and geometric reduction}
\label{sec:geometry}

We begin by developing our inferential framework.
First, we
define the score-profile divergence statistic that underpins all
subsequent results, then state the score regularity conditions and
the running policy-evaluation example
(Section~\ref{sec:setup}).
Section~\ref{sec:reduction} establishes the central geometric
reduction that connects the profile statistic to a Mahalanobis distance,
and Section~\ref{sec:levelset-theory} shows how the first-order
shape of the target functional's level set determines the limiting
distribution. 

Our inferential primitive is a triangular array of score vectors
$X_{n1}, \dots, X_{nn} \in \R^d$ obtained after sample splitting and
nuisance estimation.  In policy evaluation, the auxiliary sample is used
to fit propensity and outcome models, and the evaluation sample produces
doubly robust scores for each candidate policy (see
Example~\ref{ex:aipw} in Section~\ref{sec:setup} for details).  Given
the score array, let $\phi : \R^d \to \R$ be the target map with
$\tau_0 = \phi(\theta_0)$, and let
$f : \R_+ \to \R \cup \{+\infty\}$ be convex with
$f(1) = f'(1) = 0$ and $f''(1) = 2$.  The \emph{score-profile
  $f$-divergence statistic} is
\begin{equation}
  R_{n,f}(\tau_0)
  \defeq
  \inf_{q \in \Delta_n}
  \Bigl\{
  n D_f(q \Vert 1/n) :
  \phi\!\Bigl(\textstyle\sum_{i=1}^n q_i X_{ni}\Bigr) = \tau_0
  \Bigr\},
  \label{eq:score-profile}
\end{equation}
where $\Delta_n = \{q \ge 0 : \sum_i q_i = 1\}$ and
$D_f(q \Vert 1/n) = n^{-1} \sum_{i=1}^n f(nq_i)$.  The Euclidean case
$f_2(t) = (t-1)^2$ will play a distinguished role:
\begin{equation}
  R_{n,2}(\tau_0)
  \defeq
  \inf_{q \in \Delta_n}
  \Bigl\{
  \textstyle\sum_{i=1}^n (nq_i - 1)^2 :
  \phi\!\Bigl(\textstyle\sum_{i=1}^n q_i X_{ni}\Bigr) = \tau_0
  \Bigr\}.
  \label{eq:euclidean-profile}
\end{equation}
The statistic~\eqref{eq:score-profile} profiles out the scalar target $\tau_0$
while reweighting the score vectors, so that the inferential problem reduces
to a finite-dimensional divergence minimization on the already-computed
scores.  No nuisance model is revisited after the scores have been formed.  We
focus on the $\chi^2$ divergence to define profiles for the rest of this work
and omit the extensions to the general $f$-divergence case as it is a consequence of
tedious Taylor expansions.

\subsection{Score regularity}
\label{sec:setup}

We work directly with the score array and impose the following
high-level regularity condition.

\begin{assumption}[Score regularity]
  \label{ass:score}
  There exists a positive definite matrix $\Sigma$ such that, writing
  \[
    \bar X_n = \frac{1}{n}\sum_{i=1}^n X_{ni},
    \qquad
    \what\Sigma_n = \frac{1}{n}\sum_{i=1}^n
    (X_{ni} - \bar X_n)(X_{ni} - \bar X_n)^\top,
    \qquad
    Z_n = \sqrt{n}(\bar X_n - \theta_0),
  \]
  the following hold:
  \begin{enumerate}[label=(\roman*),leftmargin=2em]
  \item $Z_n \Rightarrow Z \sim \normal(0,\Sigma)$;
    \label{item:clt}
  \item $\what\Sigma_n \cp \Sigma$; \label{item:cov}
  \item $\max_{1 \le i \le n} \norms{X_{ni} - \bar X_n}/\sqrt{n} \cp 0$.
    \label{item:max}
  \end{enumerate}
\end{assumption}

Condition~\ref{item:max} is a Lindeberg-type negligibility requirement
that follows, for instance, from a uniform $(2+\delta)$-moment bound.
The assumption is standard in the double/debiased machine learning
literature and holds for augmented inverse-probability weighting scores
under mild overlap and rate conditions on the nuisance
estimators~\citep{ChernozhukovChDeDuHaNeRo18, AtheyWa21, KallusUe20}.

\begin{example}[Policy evaluation with doubly robust scores]
  \label{ex:aipw}
  Let $W$ denote covariates, a binary treatment assignment $A \in \{0,1\}$,
  and an observed outcome $Y \in \R$.  Write $Y^{(a)}$ for the potential
  outcome under treatment~$a$, and let $\Pi = \{\pi_1, \dots, \pi_J\}$ be a
  finite class of deterministic policies $\pi_j : \mathcal{W} \to \{0,1\}$.  The
  policy value is $V_{\pi}(P_0) = \E[Y^{(\pi(W))}]$, and the target of
  inference is the best policy value
  $\tau(P_0) = \max_{1 \le j \le J} V_{\pi_j}(P_0)$.

  Define the propensity score $e_0(w) = \P(A = 1 \mid W = w)$ and the
  conditional outcome functions
  $m_{0,a}(w) = \E[Y \mid W = w, A = a]$ for $a \in \{0,1\}$.  We
  impose the following standard conditions on the data-generating
  process and the nuisance estimators.
  \begin{enumerate}[label=(\alph*),leftmargin=2em]
  \item \emph{Unconfoundedness.}
    $(Y^{(0)}, Y^{(1)}) \perp A \mid W$.\label{item:unconf}
  \item \emph{Overlap.} There exists $\eta > 0$ such that
    $\eta \le e_0(w) \le 1 - \eta$ for $P_0$-almost every
    $w$.\label{item:overlap}
  \item \emph{Moment condition.} There exists $\delta > 0$ such that
    $\E[|Y|^{2+\delta}] < \infty$.\label{item:moment}
  \item \emph{Sample splitting.} The data are split into an auxiliary
    sample (used to fit nuisance estimators
    $\what{e}, \what{m}_0, \what{m}_1$) and an independent evaluation
    sample of size~$n$.  The evaluation-sample observations
    $(W_i, A_i, Y_i)_{i=1}^n$ are i.i.d.\ conditional on the
    auxiliary sample.\label{item:split}
  \item \emph{Rate condition.} The nuisance estimators satisfy the
    product rate
    \[
      \norms{\what{e} - e_0}_{L^2(P_0)}
      \cdot
      \max_{a \in \{0,1\}}
      \norms{\what{m}_a - m_{0,a}}_{L^2(P_0)}
      = o_p(n^{-1/2}),
    \]
    and each factor is $o_p(1)$.\label{item:rate}
  \end{enumerate}
  Conditions~\ref{item:unconf} and~\ref{item:overlap} are the
  standard identifying assumptions for causal inference from
  observational
  data~\citep{RosenbaumRu83, Imbens04}.
  Condition~\ref{item:rate} is the doubly robust rate requirement: it
  permits each nuisance function to converge slower than the
  $n^{-1/4}$ rate, so long as the product of the two rates is
  $o(n^{-1/2})$.  This accommodates modern nonparametric and
  machine-learning estimators; see~\citet{ChernozhukovChDeDuHaNeRo18}
  for sufficient conditions using cross-fitting, and~\citet{AtheyWa21}
  and~\citet{KallusUe20} for adaptations to the policy evaluation
  setting.

  Given these conditions, the evaluation-sample AIPW scores
  \[
    X_{ni,j}
    =
    \what{m}_{\pi_j(W_i)}(W_i)
    + \frac{\onevec\{A_i = \pi_j(W_i)\}}
    {\what{e}_{\pi_j(W_i)}(W_i)}
    \bigl(Y_i - \what{m}_{A_i}(W_i)\bigr),
    \qquad j = 1, \dots, J,
  \]
  satisfy Assumption~\ref{ass:score}.  The CLT
  (condition~\ref{item:clt}) and covariance consistency
  (condition~\ref{item:cov}) follow from the double robustness of the
  AIPW estimator: the product rate condition~\ref{item:rate} ensures
  that the bias from nuisance estimation is $o_p(n^{-1/2})$, so
  the scores behave asymptotically as if the true nuisance functions
  were known.  The Lindeberg condition~\ref{item:max} follows because the overlap
  bound~\ref{item:overlap} and the $(2+\delta)$-moment
  assumption~\ref{item:moment} imply
  $\E[|X_{ni,j}|^{2+\delta}] < \infty$ uniformly in~$j$,
  which is a standard sufficient condition for Lindeberg
  negligibility in the triangular-array setting.
\end{example}

\subsection{Geometric reduction}
\label{sec:reduction}

Starting with~\citet{Owen88}'s seminal work, the standard route to EL
asymptotics is by using the \emph{dual} reformulation: one writes the
score-profile statistic as a function of a Lagrange multiplier enforcing the
moment constraint, performs a stochastic Taylor expansion of the KKT
conditions around the true parameter, and inverts to obtain a $\chi^2$
limit~\citep{Owen01, HjortMcKe09}.  This dual route encounters two obstacles
for nonsmooth targets like $\tau = \max_j \theta_j$.  First, the profile over
the nonsmooth level set does not reduce to a single moment constraint with a
smooth Lagrangian, so the standard KKT expansion has no smooth analogue.
Second, even in the smooth case, a dual bootstrap proof requires a uniform
smallness condition on the resampled Lagrange multiplier---a conditional
empirical-process argument intertwined with the KKT expansion---that becomes
especially delicate when nuisance functions are refit inside each resample.  A
formal dual representation and a detailed comparison are given in
Appendix~\ref{sec:proofs-geometry}
(Proposition~\ref{prop:dual-representation}).

Instead, we take a direct \emph{primal} approach in this paper. We work
directly in the weight space of the empirical likelihood and show that the
profile statistic admits an exact geometric interpretation.  Specifically, it
equals the squared Mahalanobis distance from the score mean to the level set
of the target map.

For $t > 0$, define the \emph{local level set}
$L_t \defeq \{v \in \R^d : \phi(\theta_0 + tv) = \tau_0\}$. This is the set of
all rescaled perturbation directions along which the target functional retains
its null-hypothesized value. See Appendix~\ref{proof:exact-level-set} for the proof of the following result.
\begin{theorem}[Geometric reduction]
  \label{thm:exact-level-set}
  Under Assumption~\ref{ass:score},
  \[
    R_{n,2}(\tau_0) = d_{\what\Sigma_n}^2(Z_n, L_{1/\sqrt{n}}) + o_p(1),
  \]
  where
  $d_A^2(z,S) \defeq \inf_{v \in S} (z-v)^\top A^{-1}(z-v)$ for
  positive definite~$A$ and closed $S \subset \R^d$.
\end{theorem}

To see why this is true, consider what the profile
statistic~\eqref{eq:euclidean-profile} is actually computing.  The
problem is to find probability weights $(q_1, \dots, q_n)$ that make the
weighted mean $\sum_i q_i X_{ni}$ satisfy the target constraint
$\phi(\cdot) = \tau_0$, while staying as close as possible to the
uniform weights $q_i = 1/n$ in the squared-deviation sense
$\sum_i (nq_i - 1)^2$.  Reparameterizing the constraint by writing
the weighted mean as $\theta_0 + n^{-1/2}v$ for some direction
$v \in L_{1/\sqrt{n}}$, the problem becomes: for each direction $v$ on
the level set, what is the cheapest way to tilt the empirical weights so
that the reweighted mean lands at $\theta_0 + n^{-1/2}v$?

This tilting problem is a minimum-norm least-squares problem.  With
$a_i = nq_i - 1$ as the perturbation from uniform and
$U_{ni} = X_{ni} - \bar X_n$ as the centered scores, the constraint
becomes $\sum_i a_i U_{ni} = \sqrt{n}(v - Z_n)$ with cost
$\sum_i a_i^2$.  The minimum-norm solution is
\begin{equation}
  a_{ni}(v) = \frac{1}{\sqrt{n}}\, U_{ni}^\top \what\Sigma_n^{-1}(v - Z_n),
  \label{eq:min-norm-weights}
\end{equation}
with total cost $(v - Z_n)^\top \what\Sigma_n^{-1}(v - Z_n)$---the
squared Mahalanobis distance from the score mean $Z_n$ to the point~$v$.
The key observation is that the positivity constraints $1 + a_i \ge 0$
are asymptotically inactive: the optimal perturbations
$a_{ni}(v)$ are of order $n^{-1/2}$ uniformly over bounded $v$, so the
weights $q_i = (1 + a_i)/n$ are strictly positive for all large~$n$
(the formal statement and proof are given in
Appendix~\ref{proof:pointwise},
Proposition~\ref{prop:pointwise}).

Minimizing the Mahalanobis cost over all directions $v$ on the level set
gives Theorem~\ref{thm:exact-level-set}.  The entire argument is
deterministic once the score array is given: it is a statement about the
structure of a convex program, not about the stochastic behavior of a
dual optimizer.  No Taylor expansion of an optimized Lagrangian is
needed, and no Gaussianity of the scores is invoked.  The randomness
enters only through $Z_n$ and $\what\Sigma_n$, and the geometric
question that remains is purely about the shape of the level set
$L_{1/\sqrt{n}}$.

\subsection{Level sets to limiting distributions}
\label{sec:levelset-theory}

By Theorem~\ref{thm:exact-level-set}, the profile statistic reduces to
the squared Mahalanobis distance from $Z_n$ to the level set
$L_{1/\sqrt{n}}$.  Since $Z_n \Rightarrow Z \sim \normal(0, \Sigma)$,
the limiting distribution of $R_{n,2}(\tau_0)$ is determined by the
first-order shape to which $L_{1/\sqrt{n}}$ converges as $n$ grows.

To make this concrete, return to the motivating example.  The max
functional $\tau = \max_j \theta_j$ applied to $J$ policy values has
level set
\[
  L_t = \{v \in \R^J : \max_{1 \le j \le J}(\theta_{0j} + tv_j) = \tau_0\}.
\]
This is the set of all local perturbation directions~$v$ that preserve
the best policy value at~$\tau_0$.  For small~$t$, the suboptimal
policies---those with $\theta_{0k}$ strictly below
$\tau_0$---are irrelevant: their perturbed values
$\theta_{0k} + tv_k$ remain below $\tau_0$ as long as $v$ stays
bounded.  The constraint therefore involves only the
\emph{optimal} policies $J_0 = \arg\max_j \theta_{0j}$.

Two cases arise.  If a single policy is uniquely
optimal ($|J_0| = 1$, say $J_0 = \{j_0\}$), then the constraint reduces to
$v_{j_0} = 0$: a \emph{hyperplane} in~$\R^J$.  As we will show in subsequent
sections, projecting $Z_n$ onto a hyperplane produces a scalar ratio that
converges to $\chi_1^2$---the familiar Wilks phenomenon.  If two or more
policies are tied for optimality ($|J_0| > 1$), the constraint becomes
$\max_{j \in J_0} v_j = 0$: a \emph{cone} whose faces correspond to the
different tie-breaking patterns.  Projecting $Z_n$ onto a cone yields a
non-$\chi^2$ limit that depends on the geometry of the cone and the
covariance~$\Sigma$.

In both cases, $L_t$ converges to a well-defined limit set~$C$ as
$t \downarrow 0$, and the nature of $C$---hyperplane or
cone---determines the asymptotic distribution.  The practitioner does
not need to know which case applies: whether $|J_0|=1$ or
$|J_0|>1$ is itself an inferential question that is at least as hard as
the confidence bound being sought, so any method that requires assuming
uniqueness a priori is circular.  Our framework handles both geometries
with a single procedure by estimating $C$ from the data
(Section~\ref{sec:nonregular}).  The following assumption formalizes this
convergence for a general target functional.

\begin{assumption}[Level-set convergence]
  \label{ass:levelset}
  There exists a closed set $C \subset \R^d$ such that for every fixed
  $M < \infty$,
  $d_H(L_t \cap B_M,\, C \cap B_M) \to 0$ as $t \downarrow 0$, where
  $d_H$ denotes Hausdorff distance and
  $B_M = \{v : \norms{v} \le M\}$.
\end{assumption}

Assumption~\ref{ass:levelset} is the only structural condition the
subsequent theory requires of the target functional~$\phi$.  In
particular, the geometric reduction
(Theorem~\ref{thm:exact-level-set}), the asymptotic distribution
(Theorem~\ref{thm:general-set-limit}), and both bootstrap procedures
(Theorems~\ref{thm:hd-boot} and~\ref{thm:hdd-boot}) hold for any
$\phi$ satisfying this assumption; we develop the max functional in
detail because it is the leading case in policy evaluation and its
combinatorial structure yields closed-form profile bounds.

The role of the truncation at radius~$M$ is technical: it ensures that
convergence is stated on compacts, which suffices because the minimizing
projection in Theorem~\ref{thm:exact-level-set} lies in a bounded
region with high probability.  For the max functional, the informal
argument above can be made precise.

\begin{proposition}[Level-set structure of the max functional]
  \label{prop:max-levelset}
  Let $\phi(\theta) = \max_{1 \le j \le J} \theta_j$. If
  $\Delta_0 \defeq \min_{k \notin J_0}(\tau_0 - \theta_{0k}) > 0$
  (with $\Delta_0 = +\infty$ when $J_0 = \{1,\dots,J\}$), then for
  every fixed $M < \infty$ and every $0 < t < \Delta_0/M$,
  \[
    L_t \cap B_M = \{v \in B_M : \max_{j \in J_0} v_j = 0\}.
  \]
  Consequently:
  \begin{enumerate}[label=(\roman*),leftmargin=2em]
  \item If $|J_0| = 1$, Assumption~\ref{ass:levelset} holds with
    $C = H = \{v : v_{j_0} = 0\}$, a hyperplane.
  \item If $|J_0| > 1$, Assumption~\ref{ass:levelset} holds with
    $C = \{v : \max_{j \in J_0} v_j = 0\}$, a cone.
  \end{enumerate}
\end{proposition}
\noindent\emph{Proof.} See Appendix~\ref{proof:max-levelset}.\medskip

Assumption~\ref{ass:levelset} paired with
Theorem~\ref{thm:exact-level-set} now reduces the calibration problem
to understanding $d_\Sigma^2(Z, C)$ for the appropriate limit set~$C$.
See Appendix~\ref{proof:general-set-limit} for a formal derivation.
\begin{theorem}[Distance-to-set limit]
  \label{thm:general-set-limit}
  Under Assumptions~\ref{ass:score} and~\ref{ass:levelset},
  \[
    R_{n,2}(\tau_0) = d_{\what\Sigma_n}^2(Z_n, C) + o_p(1)
    \;\Rightarrow\;
    d_\Sigma^2(Z, C),
    \qquad Z \sim \normal(0, \Sigma).
  \]
\end{theorem}

Theorem~\ref{thm:general-set-limit} reduces the problem of calibrating
the profile statistic to identifying the first-order limit set~$C$.
When $C$ is a hyperplane, the distance $d_\Sigma^2(Z, C)$ is a scalar
quadratic form and the limit is $\chi_1^2$---the standard Wilks
calibration (Section~\ref{sec:regular}).  When $C$ is a cone, however,
the calibration problem is harder.  The distance
$d_\Sigma^2(Z, C)$ is the minimum of a quadratic form over the faces
of $C$, each face corresponding to a different subset of the tied
policies that could attain the maximum.  The resulting limit is a
weighted mixture of $\chi^2$ distributions ($\bar\chi^2$) whose mixing
probabilities depend on the unknown cone $C$ and the covariance
$\Sigma$---a non-pivotal distribution that cannot be tabulated.
Applying the standard $\chi_1^2$ critical value in this setting amounts
to projecting onto a single face (the tangent hyperplane), ignoring all
other faces of the cone, and systematically underestimates the true
critical value.  Section~\ref{sec:nonregular} develops a corrected
bootstrap that estimates the cone from the data and produces the
correct critical value without requiring the practitioner to know $J_0$.

\section{Smooth Statistical Functionals}
\label{sec:regular}

Before turning to the main results of this paper---bootstrap calibration for
nonsmooth functionals---we develop the theory for smooth targets.  The smooth
case serves an expository and illustrative purpose: it shows that the
geometric framework recovers classical results (the Wilks $\chi_1^2$ limit and
ordinary bootstrap consistency) as special cases, and it introduces the
closed-form profile bounds and the standard bootstrap approach.  Existing EL
results for the max functional~\citep{DuchiGlNa21} assume $|J_0|=1$
throughout, which, as we argued in Section~\ref{sec:intro}, is at least as
hard to verify as the inferential task itself.  Our framework will ultimately
remove this requirement (Section~\ref{sec:nonregular}), but the smooth case is
the natural starting point because it isolates the geometric mechanism.

When the target functional $\phi$ is Hadamard differentiable at $\theta_0$
with nonzero gradient $a_0 = \nabla\phi(\theta_0)$, the local level set
converges to a hyperplane $H = \{v : a_0^\top v = 0\}$.  This is the smooth
case: the profile statistic converges to $\chi_1^2$, recovering the classical
Wilks phenomenon, and an ordinary score bootstrap is consistent.  We now
develop these results and instantiate them for the policy evaluation problem.
Recall that $J_0 = \arg\max_j \theta_{0j}$ denotes the set of optimal
policies.  By Proposition~\ref{prop:max-levelset}(i), when a single policy is
uniquely optimal ($|J_0| = 1$, say $J_0 = \{j_0\}$) the level set is the
hyperplane $H = \{v : v_{j_0} = 0\}$.  We formalize this as an assumption on
the general target functional.

\begin{assumption}[Hyperplane regularity]
  \label{ass:hd}
  Assumption~\ref{ass:levelset} holds with $C = H$ for a hyperplane
  $H = \{v : a_0^\top v = 0\}$ with $a_0 \neq 0$.
\end{assumption}

\begin{corollary}[Smooth limit]
  \label{cor:hd}
  Under Assumptions~\ref{ass:score} and~\ref{ass:hd},
  \[
    R_{n,2}(\tau_0)
    =
    \frac{(a_0^\top Z_n)^2}{a_0^\top \what\Sigma_n\, a_0} + o_p(1)
    \;\Rightarrow\; \chi_1^2.
  \]
\end{corollary}
\noindent\emph{Proof.} See Appendix~\ref{proof:hd}.\medskip

The squared Mahalanobis distance from a Gaussian vector to a
hyperplane is a scalar ratio $\chi_1^2$---this is simply the statement
that the signed distance from $Z_n$ to $H$, normalized by the
appropriate variance, is asymptotically standard normal.  For the max
functional with a unique optimum at policy~$j_0$, the profile statistic
reduces to $n(\bar X_{n,j_0} - \tau_0)^2 / \what\Sigma_{n,j_0 j_0}
+ o_p(1)$: the squared $t$-statistic for the best policy's value.  The
EL machinery produces the familiar Wald-type test as a special case, but
with the important difference that the critical value has been derived
from the profile geometry rather than from a variance estimation.

\subsection{Closed-form profile bounds}

Because the max functional has combinatorial structure, the profile
statistic can be expressed in closed form.  The closed-form expressions
show how the profile statistic depends on the policy-value gaps and
covariance structure, and they serve as the foundation for the simplex
lower-bound formula used in both the smooth and nonsmooth cases.

The max functional's combinatorial structure allows the geometric
objects in Theorem~\ref{thm:exact-level-set} to be computed in closed
form.  By Theorem~\ref{thm:exact-level-set}, the Euclidean profile
statistic for the max target is
\begin{equation}
  R_{n,2}^{\max}(\tau)
  = n\, d_{\what\Sigma_n}^2(\bar X_n, M_\tau) + o_p(1),
  \label{eq:max-profile-stat}
\end{equation}
where $M_\tau = \{m \in \R^J : \max_j m_j = \tau\}$.
For the decision-maker's primary object of interest---a lower confidence
bound on the best policy value---the profile statistic admits a
particularly clean dual form.

\begin{proposition}[Lower bound via simplex optimization]
  \label{prop:max-lower}
  For a critical value $c > 0$, define the acceptance ellipsoid
  $E_n(c) = \{m \in \R^J : n(m - \bar X_n)^\top \what\Sigma_n^{-1}
  (m - \bar X_n) \le c\}$.  The lower endpoint
  $L_n(c) \defeq \inf_{m \in E_n(c)} \max_j m_j$ satisfies
  \[
    L_n(c)
    = \max_{w \in \Delta_J}
    \left\{
      w^\top \bar X_n
      - \sqrt{\frac{c}{n}}\,\sqrt{w^\top \what\Sigma_n\, w}
    \right\},
    \qquad
    \Delta_J = \{w \ge 0,\; w^\top \onevec = 1\}.
  \]
\end{proposition}
\noindent This result is known~\cite{DuchiGlNa21}, and we rederive it for completeness
 in Appendix~\ref{proof:max-lower}.
The simplex representation is computationally attractive: after the
score mean and covariance have been computed, the lower bound is obtained
by solving a $J$-dimensional convex program rather than manipulating
$n$ empirical-likelihood weights.  The optimal weight vector~$w^*$ is
concentrated on the near-optimal policies, so in practice the program
effectively reduces to a problem of dimension~$|J_0|$.

The leading-order cost of the projected-joint workaround discussed in
Section~\ref{sec:intro} can also be made precise.  Recall that
$\Delta_0 = \min_{k \notin J_0}(\tau_0 - \theta_{0k})$ is the gap
between the best policy value and the second-best.

\begin{proposition}[Inflation factor]
  \label{prop:unique-inflation}
  Suppose $|J_0| = 1$ with $\Delta_0 > 0$.
  Let $c_{1,\alpha} = \chi^2_{1,1-\alpha}$ and
  $c_{J,\alpha} = \chi^2_{J,1-\alpha}$.  Then the direct profile lower
  bound satisfies
  \[
    L_{n,\mathrm{dir}}(1-\alpha)
    = \bar X_{n,j_0}
    - \sqrt{\frac{c_{1,\alpha}}{n}}\,\sqrt{\what\Sigma_{n,j_0 j_0}}
    + o_p(n^{-1/2}),
  \]
  while the projected joint lower bound replaces $c_{1,\alpha}$ with
  $c_{J,\alpha}$.  The radius ratio converges to
  $\sqrt{c_{J,\alpha}/c_{1,\alpha}}$.
\end{proposition}
See Appendix~\ref{proof:unique-inflation} for the derivation. At the $95\%$
level with $J = 20$, this ratio is approximately $2.86$: the projected joint
method produces a bound nearly three times wider than necessary, paying a
$J$-dimensional penalty for a scalar target.  The gap grows as $\sqrt{J}$ and
becomes severe in large policy classes.

\subsection{Ordinary score bootstrap}

The $\chi_1^2$ calibration from Corollary~\ref{cor:hd} is valid when
the level set is exactly a hyperplane, but the $\chi_1^2$ critical
value may be slightly miscalibrated in finite samples.  Bootstrap
calibration can improve the critical value, and the primal reduction
makes this especially simple: because the profile statistic operates
entirely at the score level, the bootstrap can be applied to the
frozen score vectors with no nuisance refitting.

Let $X_{n1}^*, \dots, X_{nn}^*$ be an Efron bootstrap resample from the
empirical distribution of $\{X_{ni}\}$, and set
$Z_n^* = \sqrt{n}(\bar X_n^* - \bar X_n)$.

\begin{assumption}[Score bootstrap validity]
  \label{ass:hdboot}
  Conditionally on the data,
  $\sup_{g \in \mathrm{BL}_1(\R^d)}
  |\E^*[g(Z_n^*)] - \E[g(Z)]| \cp 0$ for
  $Z \sim \normal(0, \Sigma)$.  Moreover there exists a random
  hyperplane $\what H_n$ with
  $d_H(\what H_n \cap B_M, H \cap B_M) \cp 0$ for every fixed~$M$.
\end{assumption}

In the finite-policy setting with a unique optimum,
$\what H_n = \{v : v_{\what j} = 0\}$ where
$\what j = \arg\max_j \bar X_{n,j}$.  Assumption~\ref{ass:hdboot} then
requires only that $\what j = j_0$ eventually---i.e., that the
sample correctly identifies the best policy---which holds with
probability tending to one whenever $\Delta_0 > 0$.

We write $\mathcal L(X)$ for the law (probability distribution) of a
random variable~$X$, $\mathcal L^*(X \mid \text{data})$ for the
conditional law under the bootstrap measure, and
$\Rightarrow_p$ for convergence in distribution in outer
probability.

\begin{theorem}[Ordinary score bootstrap]
  \label{thm:hd-boot}
  Under Assumptions~\ref{ass:score},~\ref{ass:hd},
  and~\ref{ass:hdboot}, the statistic
  $T_n^* \defeq d_{\what\Sigma_n}^2(Z_n^*, \what H_n)$ satisfies
  \[
    \mathcal L^*(T_n^* \mid X_{n1}, \dots, X_{nn})
    \cd
    \chi_1^2.
  \]
  If $R_{n,2}^*(\what\tau_n) = T_n^* + o_{\P^*}(1)$ with
  $\what\tau_n = \phi(\bar X_n)$, then $R_{n,2}^*(\what\tau_n)$ is also
  a consistent bootstrap approximation to $R_{n,2}(\tau_0)$.
\end{theorem}
See Appendix~\ref{proof:hd-boot} for the proof. Compared with the
semiparametric EL bootstrap of~\citet{HjortMcKe09}, the bootstrap here
operates on the profile divergence for the target functional itself, not on an
estimating-equation EL for the vector parameter.  Crucially, every bootstrap
draw is $O(nJ)$: it resamples the score vectors, computes a new mean, and
evaluates a Mahalanobis distance.  No nuisance model is revisited inside the
bootstrap loop!


\section{Corrected bootstrap for nonsmooth functionals}
\label{sec:nonregular}

We now turn to the central challenge that motivates this paper.  When
two or more policies are tied for optimality, the max functional is only
Hadamard directionally differentiable (HDD) at $\theta_0$, and the
local level set converges to a cone rather than a hyperplane.  The
geometry changes: the profile statistic no longer converges
to $\chi_1^2$, the ordinary bootstrap is inconsistent, and a corrected
multiplier bootstrap is needed.  As argued in
Section~\ref{sec:intro}, the uniqueness assumption $|J_0|=1$ required by
existing EL results~\citep{DuchiGlNa21} is at least as hard to verify as
the inference itself, and the experiments in Section~\ref{sec:sim} show
that naively assuming uniqueness leads to severe undercoverage.  This
section develops the theory and implementation of a procedure that adapts
to the unknown geometry automatically.

Recall that $J_0 = \arg\max_j \theta_{0j}$ is the set of optimal
policies and
$\Delta_0 = \min_{k \notin J_0}(\tau_0 - \theta_{0k})$ is the
optimality gap.
By Proposition~\ref{prop:max-levelset}(ii), when the optimal policy set
has $|J_0| > 1$, the level set converges to the cone
$C = \{v : \max_{j \in J_0} v_j = 0\}$.
This cone has $|J_0|$ faces, one for each way of perturbing the policy
values so that a particular subset of the tied policies attains the
maximum.  Near a two-way tie, $C$ is a wedge in $\R^J$; near a
three-way tie, it is the intersection of three half-spaces, and so on.

\begin{assumption}[Cone regularity]
  \label{ass:hdd}
  Assumption~\ref{ass:levelset} holds with
  $C = \{v \in \R^d : \phi'_{\theta_0}(v) = 0\}$, where
  $\phi'_{\theta_0}$ is the directional derivative of~$\phi$
  at~$\theta_0$.
\end{assumption}

For the max functional at a tie, Assumption~\ref{ass:hdd} holds with
$C = \{v : \max_{j \in J_0} v_j = 0\}$
(Proposition~\ref{prop:max-levelset}).

\begin{corollary}[Distance-to-cone limit]
  \label{cor:hdd}
  Under Assumptions~\ref{ass:score} and~\ref{ass:hdd},
  \[
    R_{n,2}(\tau_0) = d_{\what\Sigma_n}^2(Z_n, C) + o_p(1)
    \;\Rightarrow\; d_\Sigma^2(Z, C),
    \qquad Z \sim \normal(0, \Sigma).
  \]
\end{corollary}

The limit $d_\Sigma^2(Z, C)$ is the squared Mahalanobis distance from a
Gaussian vector to a cone---a weighted $\bar\chi^2$ distribution whose
mixing probabilities depend on the unknown cone $C$ and the
covariance $\Sigma$.  Unlike the hyperplane case, this distribution is
non-pivotal: it depends on the geometry of the tie through $J_0$ and
on $\Sigma$.  Critical values must therefore be estimated from the data.

To see why the standard calibration fails, consider what the $\chi_1^2$
critical value actually assumes.  In the smooth case
(Section~\ref{sec:regular}), the level set is a hyperplane, and the
Mahalanobis projection lands on a single linear subspace---there is
only one ``face'' to project onto, yielding a scalar ratio.  At a cone,
the projection must search over $2^{|J_0|}-1$ faces, one for each
nonempty subset of the tied policies that could attain the maximum.
The $\chi_1^2$ calibration corresponds to projecting onto the tangent
hyperplane of the cone---effectively picking a single face and ignoring
all others.  This is incorrect because the true projection may land on
a different face (or on an edge between faces), and the $\bar\chi^2$
mixing probabilities that weight these events depend on the unknown
geometry.  Neither the number of faces nor their orientations relative
to~$\Sigma$ can be read off from the data without estimating $J_0$.

The ordinary bootstrap cannot provide this estimate.  At any finite
sample, the estimated best policy
$\what j = \arg\max_j \bar X_{n,j}$ is generically unique (ties occur
with probability zero under continuous distributions), so the bootstrap
always sees a unique optimum and estimates the level set as a
hyperplane.  The resulting bootstrap distribution converges to
$\chi_1^2$, regardless of whether the true level set is a hyperplane or
a cone.  When two or more policies are truly tied, the bootstrap
systematically underestimates the critical value and undercovers.  This
is the well-known bootstrap inconsistency for nonsmooth functionals
identified by~\citet{FangSa19}: the ordinary bootstrap linearizes the
directional derivative, collapsing the cone to a tangent hyperplane.

\subsection{Face structure of the profile statistic}

The geometric reduction of Theorem~\ref{thm:exact-level-set} converted
the profile statistic into a Mahalanobis projection; the level-set
analysis of Section~\ref{sec:levelset-theory} identified the limit set
as a cone.  We now make explicit the combinatorial structure of this
projection---a step that is new to this paper and has no analogue in
existing EL or corrected-bootstrap theory.  The face decomposition
below shows exactly how the profile statistic decomposes over the faces
of the cone, revealing which subsets of policies are statistically
competitive for optimality and how the covariance structure governs
the projection cost on each face.

The level set $M_\tau = \{m \in \R^J : \max_j m_j = \tau\}$ is the
union of faces indexed by active sets $A \subseteq \{1,\dots,J\}$, and
the projection onto $M_\tau$ decomposes accordingly.

\begin{proposition}[Face decomposition]
  \label{prop:max-face}
  For a nonempty $A \subseteq \{1,\dots,J\}$, let $C_A$ extract the
  coordinates in~$A$ and define
  \[
    \Pi_A(\tau) = \bar X_n
    - \what\Sigma_n C_A^\top (C_A \what\Sigma_n C_A^\top)^{-1}
    (C_A \bar X_n - \tau\onevec).
  \]
  If $\Pi_A(\tau)_k \le \tau$ for all $k \notin A$, the projection cost
  onto the face $\{m : m_j = \tau,\; j \in A\}$ is
  \[
    R_{n,A}^{\max}(\tau)
    = n\,(C_A \bar X_n - \tau\onevec)^\top
    (C_A \what\Sigma_n C_A^\top)^{-1}
    (C_A \bar X_n - \tau\onevec),
  \]
  and the profile statistic is
  $R_{n,2}^{\max}(\tau) = \min_{A :\, \Pi_A(\tau)_{A^c} \le \tau}
  R_{n,A}^{\max}(\tau) + o_p(1)$.
\end{proposition}
\noindent\emph{Proof.} See Appendix~\ref{proof:max-face}.\medskip

The face decomposition complements the simplex lower bound in
Proposition~\ref{prop:max-lower} and makes precise the source of the
calibration difficulty discussed above.  While the simplex
representation directly yields a lower confidence bound on the optimal
value, the face decomposition reveals which active set---and hence
which subset of tied policies---achieves the minimum projection cost.
When the optimum is unique ($|J_0| = 1$), there is a single face and
the simplex bound recovers the familiar scalar $t$-statistic
(Proposition~\ref{prop:unique-inflation})---the standard $\chi_1^2$
calibration is correct in this case because the projection has no
combinatorial structure.  When multiple policies are tied, the profile
statistic searches over $2^{|J_0|} - 1$ faces of the cone: this is the
combinatorial complexity that the $\chi_1^2$ calibration ignores and
that makes the corrected bootstrap essential.  The minimizing face
identifies which policies the data treats as near-optimal, providing
interpretable diagnostic information alongside the confidence bound.

\subsection{Corrected multiplier bootstrap}

The corrected bootstrap requires an estimate of the cone $C$, which in
turn requires identifying the set of near-optimal policies.  We estimate
this set as
\[
  \what J_n = \Bigl\{
  j : \bar X_{n,j} \ge \max_k \bar X_{n,k} - \kappa_n
  \Bigr\},
  \qquad
  \kappa_n
  = \max_{1 \le j \le J}
  \sqrt{\frac{\what\Sigma_{n,jj}\,\log n}{n}},
\]
and form the estimated cone
$\what C_n = \{v : \max_{j \in \what J_n} v_j = 0\}$.  The threshold
$\kappa_n$ is chosen to decay to zero slowly enough that
$\sqrt{n}\,\kappa_n \to \infty$: this ensures that truly optimal
policies are included in $\what J_n$ (their sample gap from the leader
is $O_p(n^{-1/2}) = o_p(\kappa_n)$), while truly suboptimal policies are
excluded eventually (their gap converges to a positive constant
$\Delta_0 > \kappa_n$).

\begin{proposition}[Validity of the active-set estimator]
  \label{prop:active-set}
  Under Assumption~\ref{ass:score}, the threshold
  $\kappa_n$ satisfies $\kappa_n \cp 0$ and
  $\sqrt{n}\,\kappa_n \cp \infty$.  Consequently:
  \begin{enumerate}[label=(\roman*),leftmargin=2em]
  \item \emph{(Selection consistency.)} With probability tending to one,
    $\what J_n \supseteq J_0$.  If not all policies are optimal
    (i.e., $J_0 \subsetneq \{1,\dots,J\}$, so that
    $\Delta_0 = \min_{k \notin J_0}(\tau_0 - \theta_{0k}) > 0$), then
    $\P(\what J_n = J_0) \to 1$.
  \item \emph{(Cone estimation.)} For every fixed $M < \infty$,
    $d_H(\what C_n \cap B_M,\, C \cap B_M) \cp 0$,
    where $C = \{v : \max_{j \in J_0} v_j = 0\}$.  In particular,
    $\what C_n$ satisfies the cone-estimation requirement in
    Assumption~\ref{ass:mult}.
  \end{enumerate}
\end{proposition}
\noindent\emph{Proof.} See Appendix~\ref{proof:active-set}.\medskip

With the estimated cone in hand, we calibrate the profile statistic
using a multiplier bootstrap.  Let $\xi_1, \dots, \xi_n$ be $\simiid$
multipliers independent of the data with $\E\xi_i = 0$,
$\E\xi_i^2 = 1$, $\E|\xi_i|^{2+\eta} < \infty$, and set
$Z_n^\xi = n^{-1/2}\sum_{i=1}^n \xi_i(X_{ni} - \bar X_n)$.

\begin{assumption}[Multiplier bootstrap and cone estimation]
  \label{ass:mult}
  Conditionally on the data,
  \[\sup_{g \in \mathrm{BL}_1(\R^d)}
    |\E_\xi[g(Z_n^\xi)] - \E[g(Z)]| \cp 0
  \] for
  $Z \sim \normal(0, \Sigma)$.  Moreover there exists a random closed
  set $\what C_n$ with
  $d_H(\what C_n \cap B_M, C \cap B_M) \cp 0$ for every fixed~$M$.
\end{assumption}

By Proposition~\ref{prop:active-set}, the active-set estimator
$\what C_n = \{v : \max_{j \in \what J_n} v_j = 0\}$ satisfies the
cone-estimation requirement.
Analogously to Section~\ref{sec:regular}, we write
$\mathcal L_\xi(\cdot \mid \text{data})$ for the conditional law
under the multiplier distribution.

\begin{theorem}[Corrected multiplier bootstrap]
  \label{thm:hdd-boot}
  Under Assumptions~\ref{ass:score},~\ref{ass:hdd},
  and~\ref{ass:mult}, the statistic
  $T_n^\xi \defeq d_{\what\Sigma_n}^2(Z_n^\xi, \what C_n)$ satisfies
  \[
    \mathcal L_\xi(T_n^\xi \mid X_{n1}, \dots, X_{nn})
    \;\Rightarrow_p\;
    \mathcal L\!\bigl(d_\Sigma^2(Z, C)\bigr).
  \]
  If the cdf of $d_\Sigma^2(Z,C)$ is continuous, the convergence holds
  uniformly in the Kolmogorov metric.
\end{theorem}
\noindent\emph{Proof.} See Appendix~\ref{proof:hdd-boot}.\medskip

Theorem~\ref{thm:hdd-boot} is the profile-EL counterpart of the
Fang--Santos corrected bootstrap~\citep{FangSa19}.  The local map
$z \mapsto d_\Sigma^2(z, C)$ is not imposed externally but arises from
the profile-divergence geometry itself.

\paragraph{Summary of the inferential procedure}

The corrected bootstrap critical value is the empirical $(1-\alpha)$
quantile of
$T_n^{\xi,(b)} = d_{\what\Sigma_n}^2(Z_n^{\xi,(b)}, \what C_n)$ over
$b = 1, \dots, B$ multiplier draws.  Plugging
$c = \what c_{1-\alpha}$ into Proposition~\ref{prop:max-lower} yields
the final lower confidence bound.

The entire procedure, after the initial nuisance fit, proceeds in five
steps: (1)~form the AIPW score vectors on the evaluation sample;
(2)~compute the sample mean $\bar X_n$ and covariance $\what\Sigma_n$;
(3)~estimate the active set $\what J_n$ via the threshold $\kappa_n$,
treating the problem as smooth if $|\what J_n| = 1$ and nonsmooth
otherwise; (4)~calibrate using $\chi_1^2$ or the ordinary score
bootstrap in the smooth case (Theorem~\ref{thm:hd-boot}), or the
corrected multiplier bootstrap in the nonsmooth case; and (5)~solve
the $J$-simplex program in Proposition~\ref{prop:max-lower} with the
resulting critical value.  The bootstrap in step~(4) operates entirely
on the score vectors---each draw is $O(nJ)$ with no reference to
the nuisance models---and the simplex program in step~(5) is
effectively of dimension~$|\what J_n|$, because the optimal weight
concentrates on the near-optimal policies.


\section{Simulation study}
\label{sec:sim}

We present four sets of experiments.  The first two are score-level
simulations that isolate the geometric advantage of direct profiling
over the projected joint alternative.  The third validates the method
end-to-end in a semiparametric policy-evaluation pipeline with
cross-fitted nuisance models.  The fourth compares our method against
the non-EL corrected bootstrap of~\citet{FangSa19}.

Throughout, we construct one-sided 95\% lower confidence bounds for
$\tau_0 = \max_j \theta_{0j}$ and report empirical coverage and the average
shortfall $\tau_0 - L_n$ (smaller is tighter, subject to coverage).  The three
methods are:
\begin{enumerate}[label=(\alph*),leftmargin=1.5em]
\item \textbf{Ours:} the lower bound from
  Proposition~\ref{prop:max-lower} with the corrected multiplier
  bootstrap of Theorem~\ref{thm:hdd-boot}.
\item \textbf{Projected joint EL:} the same simplex lower-bound
  formula but with the ambient critical value $\chi^2_{J,0.95}$.
  This is the bound obtained by inverting a joint $\chi^2_J$
  confidence region and projecting onto the max functional---the
  standard workaround when a direct profile is unavailable.
\item \textbf{Selected-policy Wald:} choose
  $\what j = \arg\max_j \bar X_{n,j}$ and report
  $\bar X_{n,\what j} - z_{0.95}\sqrt{\what\Sigma_{n,\what j\what j}/n}$.
\end{enumerate}
All score-level Monte Carlo cells use 1{,}000 repetitions and
1{,}000 multiplier draws; the semiparametric experiment uses
150~repetitions. 

\subsection{The dimension penalty of projected joint EL}
\label{sec:sim-dimension}

Proposition~\ref{prop:unique-inflation} predicts that, at a unique
optimum, the projected joint approach inflates the lower-bound
radius by a factor of $\sqrt{\chi^2_{J,0.95}/\chi^2_{1,0.95}}$, which
grows as roughly $\sqrt J$.  We test this prediction directly.

We fix $n = 500$ and vary $J \in \{5, 10, 20, 50, 100\}$.  The
mean vector has a unique optimum at $\theta_{0,1} = 0.35$ with all
other coordinates strictly below.  Scores are generated as
\[
  X_{ni} = \theta_0 + 0.70\,G_i + 0.20\,E_i + 0.10\,S_i b,
\]
where $G_i$ is mean-zero Gaussian with covariance
$\Sigma_{jk} = (0.5^{|j-k|} + 0.15)\sqrt{v_j v_k} /
(0.5^0 + 0.15)$, $v_j = 1 - 0.3(j-1)/(J-1)$,
the vector $E_i$ has independent standardized~$t_5$ coordinates,
$S_i \sim \mathrm{Exp}(1) - 1$ is a scalar, and
$b_j = 1 - 0.5(j-1)/(J-1)$.  The mixture of Gaussian, heavy-tailed,
and skewed components ensures the scores are not Gaussian.

\begin{figure}[ht]
  \centering
  \includegraphics[width=\textwidth]{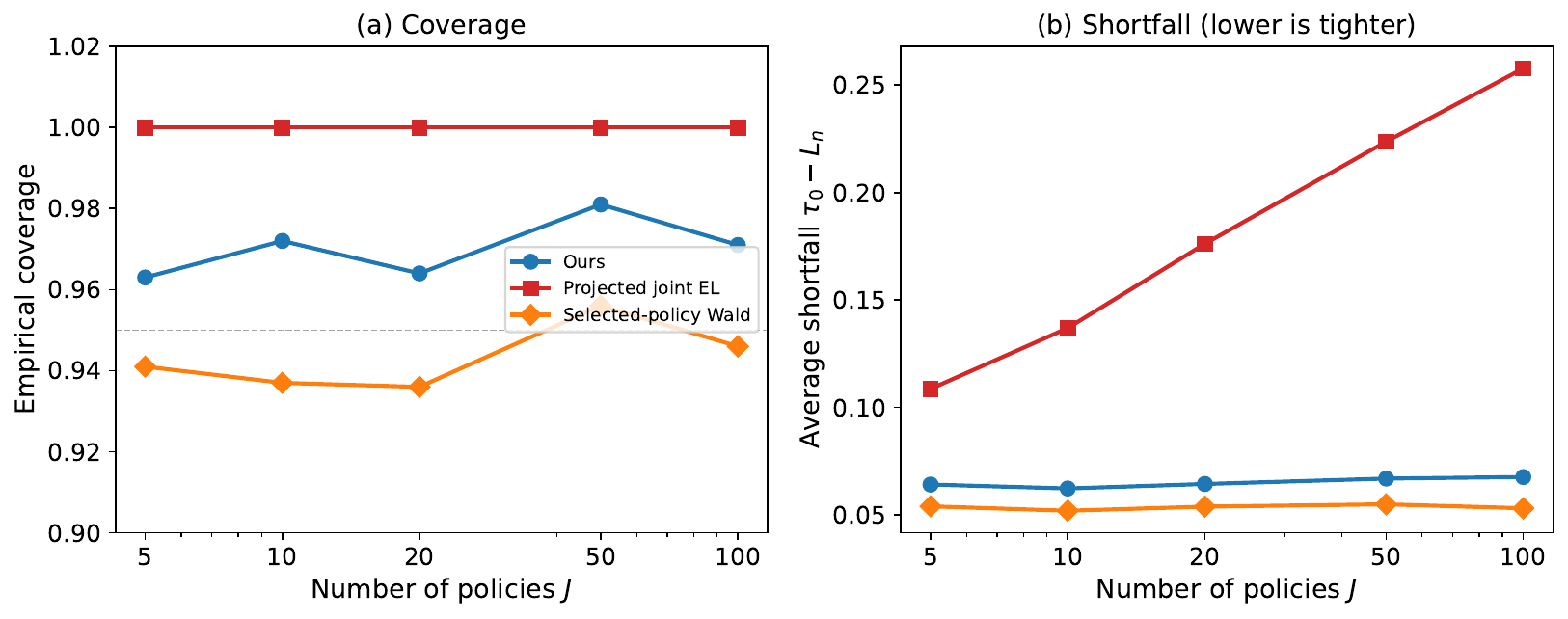}
  \caption{Coverage and average shortfall of one-sided 95\% lower
    bounds as a function of $J$, with $n = 500$ and a unique optimum.
    The projected joint approach pays a dimension penalty that
    grows as $\sqrt J$.  Our method avoids this entirely.}
  \label{fig:dimension-penalty}
\end{figure}

Figure~\ref{fig:dimension-penalty} confirms the theoretical prediction.
The projected joint shortfall grows approximately as $\sqrt J$: at
$J = 100$, its lower bound is roughly four times wider than ours.  Its
coverage is always~$1$, reflecting systematic overconservatism---it
pays a $J$-dimensional critical value for a scalar target.  Our method
maintains near-nominal coverage across all values of~$J$, with
shortfall that barely changes because the profile critical value
converges to~$\chi^2_1$.  

\subsection{Robustness at ties}
\label{sec:sim-ties}

The second experiment uses the same score-level DGP as
Section~\ref{sec:sim-dimension} but examines what happens when
multiple policies share the optimal value.  We fix $J = 20$ and vary
the tie multiplicity $k \in \{1, 2, 4, 8\}$ at sample sizes
$n \in \{500, 1{,}000\}$.

\begin{figure}[ht]
  \centering
  \includegraphics[width=\textwidth]{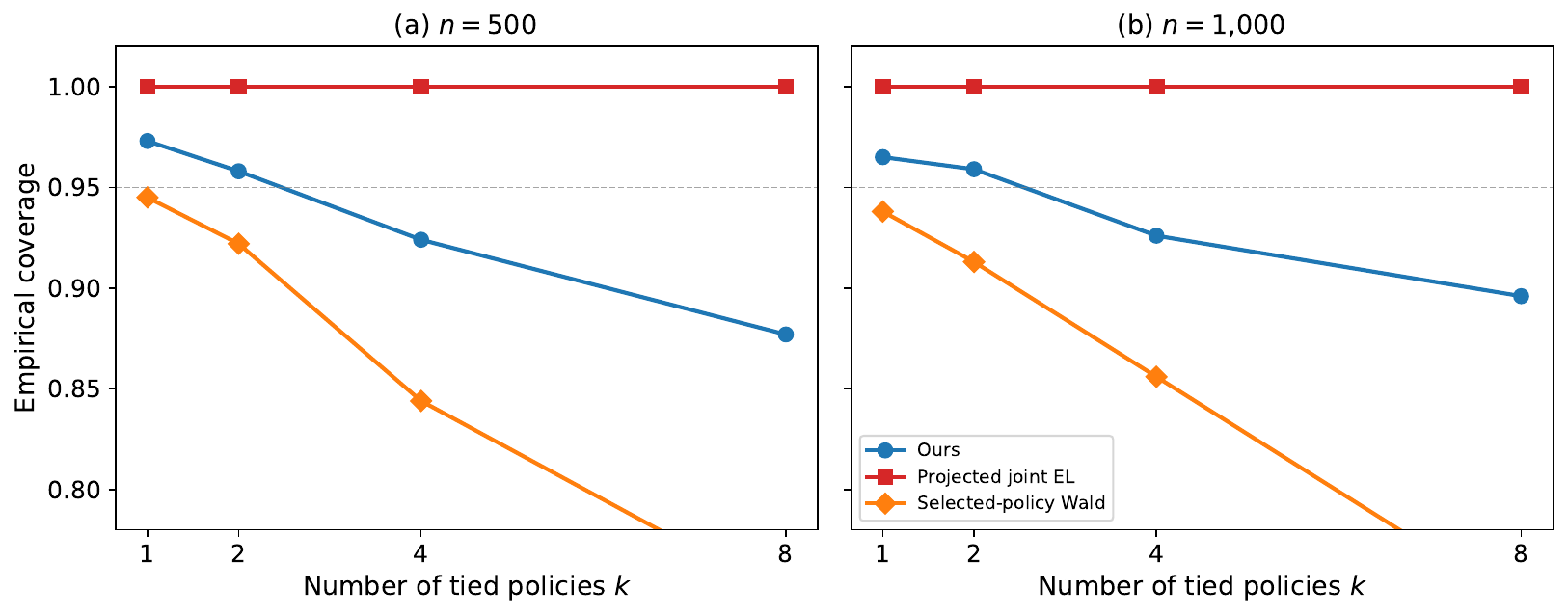}
  \caption{Coverage of one-sided 95\% lower bounds as a function of
    tie multiplicity~$k$ ($J = 20$, $n = 500$ and $n = 1{,}000$).
    The selected-policy method undercovers sharply at ties; the
    projected joint EL method overcovers regardless; our method
    maintains near-nominal coverage by adapting to the cone geometry.}
  \label{fig:tie-robustness}
\end{figure}

Figure~\ref{fig:tie-robustness} shows the expected separation.  The
selected-policy Wald bound degrades substantially as ties increase,
dropping well below~$90\%$ coverage at $k = 8$: it treats the
policy-selection step as deterministic, which fails at exactly the
decision-relevant boundary where the choice is ambiguous.  The
projected joint EL approach maintains
100\% coverage regardless, because its $J$-dimensional critical value
is so conservative that even the nonsmooth geometry cannot erode it.

Our method adapts: the estimated active set $\what J_n$ detects the
near-optimal policies, and the distance-to-cone bootstrap
(Theorem~\ref{thm:hdd-boot}) produces the correct nonstandard critical
value.  Coverage stays near the nominal level across all tie
configurations, and improves as $n$ grows.  This is the practical
payoff of the primal geometric reduction: the same procedure works at
smooth and nonsmooth points without the user having to know which
regime applies.

\subsection{Semiparametric policy evaluation}
\label{sec:sim-semi}

The third experiment validates the full pipeline.  We generate
observational data $(W, A, Y)$ with $W \sim N(0, I_6)$, propensity
$e(w) = \mathrm{expit}(0.6 w_1 - 0.5 w_2 + 0.3 w_3 w_4 -
0.2(w_5^2 - 1) + 0.15\sin w_6)$ clipped to $[0.1, 0.9]$, baseline
outcome
$\mu_0(w) = 0.5 w_1 - 0.3 w_2 + 0.2 w_3^2 - 0.15 w_4 w_5 +
0.2\cos w_6$, treatment effect
$\tau(w) = 0.6\sin w_1 + 0.4 \cdot \mathbf{1}[w_2 > 0] - 0.3 w_3 +
0.2 w_1 w_2$, and
$Y = \mu_0(W) + A\,\tau(W) + \varepsilon$,
$\varepsilon \sim N(0,1)$.  The $J = 20$ policies are linear threshold
rules $\pi_j(w) = \mathbf{1}[\beta_j^\top w_{1:3} + b_j > 0]$.
Nuisance functions are estimated with
random forests via two-fold cross-fitting, producing AIPW scores to
which all three methods are applied.  True policy values are computed
by Monte Carlo integration over the known DGP.

We construct the policy class to reflect the practically relevant
case where personalization effects are modest: five policies have
true values within~$0.01$ of the optimum (created by perturbing the
weight vector of the best policy), with the remaining~15 drawn with
attenuated random weights so they are clearly suboptimal.  This near-tied
regime arises naturally whenever several candidate treatment rules
target similar subpopulations---the typical situation in which a
practitioner most needs reliable inference on the best attainable
value.

\begin{table}[ht]
  \centering
  \caption{Coverage and average shortfall for the semiparametric
    policy-evaluation experiment ($J = 20$, five near-tied optimal
    policies, cross-fitted AIPW scores, 150 Monte Carlo repetitions).
    The selected-policy Wald method undercovers at every sample size;
    our method maintains valid coverage while achieving bounds
    roughly~$3\times$ tighter than projected joint EL.}
  \label{tab:semi}
  \begin{tabular}{l@{\hskip 1.5em}cc@{\hskip 1.5em}cc@{\hskip 1.5em}cc}
    \toprule
    & \multicolumn{2}{c}{$n = 500$}
    & \multicolumn{2}{c}{$n = 1{,}000$}
    & \multicolumn{2}{c}{$n = 2{,}000$} \\
    \cmidrule(lr){2-3} \cmidrule(lr){4-5} \cmidrule(lr){6-7}
    Method & Cov. & Shortfall
    & Cov. & Shortfall
    & Cov. & Shortfall \\
    \midrule
    Ours
      & 0.973 & 0.151
      & 0.993 & 0.099
      & 0.960 & 0.069 \\
    Projected joint EL
      & 1.000 & 0.419
      & 1.000 & 0.293
      & 1.000 & 0.206 \\
    Selected-policy Wald
      & 0.913 & 0.110
      & 0.940 & 0.073
      & 0.907 & 0.054 \\
    \bottomrule
  \end{tabular}
\end{table}

Table~\ref{tab:semi} confirms that the score-level theory carries over
to the semiparametric setting.  The selected-policy Wald method
undercovers at every sample size: its coverage is $91.3\%$ at
$n = 500$, $94.0\%$ at $n = 1{,}000$, and $90.7\%$ at
$n = 2{,}000$---well below the nominal~$95\%$ throughout.
This occurs because it treats the
policy-selection step as deterministic, ignoring the selection
uncertainty among the near-optimal candidates.  This is the failure
mode discussed in Section~\ref{sec:intro}: the uniqueness assumption
$|J_0| = 1$ is unverifiable a priori, yet methods that rely on it
miscalibrate when it fails.

Our method maintains valid coverage at every sample size
($97.3\%$, $99.3\%$, and $96.0\%$) by detecting the near-optimal
policies through the estimated active set~$\what J_n$ and calibrating
the critical value with the distance-to-cone bootstrap
(Theorem~\ref{thm:hdd-boot}).
The projected joint method maintains $100\%$ coverage throughout, but at a
shortfall penalty exceeding~$3\times$, closely matching the
theoretical prediction
$\sqrt{\chi^2_{20,0.95}/\chi^2_{1,0.95}} \approx 2.9$ from
Proposition~\ref{prop:unique-inflation}.
Our method is the only one that simultaneously achieves valid coverage
and competitive tightness in the near-tied regime.

\paragraph{Computational advantage.}
The statistical comparison above is against the projected joint $\chi^2_J$
bound, which uses the same score-level computation as our method and therefore
has similar runtime.  To measure the \emph{computational} advantage of the
score-level approach, we run a separate timing experiment comparing our method
(fit nuisances once, score-level bootstrap with $B = 1{,}000$ draws) against a
refit-each-resample plug-in bootstrap that refits all nuisance models
$(\what e, \what m_0, \what m_1)$ inside each of the same $B = 1{,}000$
resamples.  On a comparable DGP with $J = 30$ and $n = 2{,}000$, the refit
approach is roughly $460\times$ slower: each refit iteration re-trains three
random forest models, while our score-level bootstrap requires only a single
matrix multiplication per draw.  This is the algorithmic payoff of operating
at the score level: once the AIPW scores have been formed, our bootstrap never
revisits the nuisance models.  Any semiparametric EL method that couples the
bootstrap to the nuisance-estimation step---including a full implementation of
\citet{HjortMcKe09}---would incur a comparable refit cost.

\subsection{Comparison with non-EL methods}
\label{sec:sim-fang-santos}

The corrected multiplier bootstrap of~\citet{FangSa19} is the
leading non-EL method for HDD targets.  Both methods use the same
active-set and cone information, so their coverage and shortfall are
comparable when the tied policies are highly correlated.  The
structural difference is in how the lower bound is formed.
The~\citet{FangSa19} bound is
$L_n^{\mathrm{FS}} = \max_j \bar X_{n,j} -
\what q_{1-\alpha}/\sqrt n$,
which is anchored to the pointwise maximum.
Ours is
$L_n = \max_{w \in \Delta_J}\{w^\top \bar X_n - \sqrt{c/n}\,
\sqrt{w^\top\what\Sigma_n w}\}$
(Proposition~\ref{prop:max-lower}), which optimizes over
\emph{policy mixtures}.  When several near-optimal policies have low
pairwise correlation, the simplex optimizer finds a mixture~$w^*$ that
maintains a high mean while substantially reducing the variance term
$w^{*\top}\what\Sigma_n w^*$.  This is a portfolio diversification
effect that the pointwise-max formulation cannot exploit.

We isolate this effect with $J = 10$ policies, the top $k = 3$ tied
at~$0.30$, equi-correlation~$\rho$ among the tied group, and
$n = 2{,}000$.  When $\rho$ is large, the tied policies move together
and diversification is impossible.  When $\rho$ is small, a uniform
mixture over the $k$ tied policies has variance
$(1 + (k-1)\rho)/k$, which can be much smaller than any single-policy
variance.  We use $5{,}000$ Monte Carlo repetitions for this
experiment to produce smooth coverage curves.

\begin{figure}[ht]
  \centering
  \includegraphics[width=\textwidth]{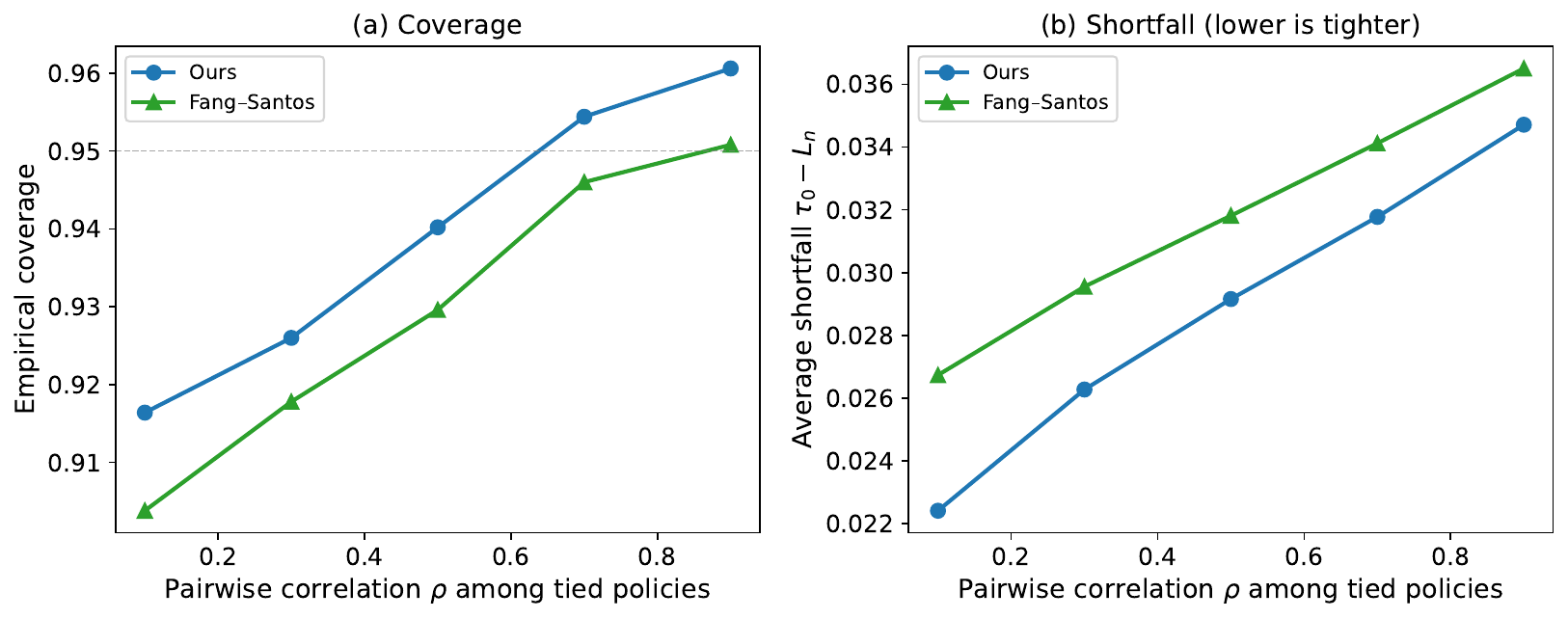}
  \caption{Our method vs.\ \citet{FangSa19} as the pairwise
    correlation~$\rho$ among $k = 3$ tied policies varies ($J = 10$,
    $n = 2{,}000$, $5{,}000$ repetitions).  (a)~Both methods maintain
    near-nominal coverage across the range; our method matches or
    exceeds~\citet{FangSa19} coverage throughout.  (b)~Our shortfall
    drops sharply as~$\rho$ decreases because the simplex optimizer
    diversifies across tied policies; the~\citet{FangSa19} shortfall
    is relatively flat.}
  \label{fig:fang-santos}
\end{figure}

Figure~\ref{fig:fang-santos} confirms the diversification mechanism.
Both methods maintain near-nominal coverage across the full range
of~$\rho$, with our method matching or exceeding~\citet{FangSa19}
coverage at every value.
At $\rho \ge 0.7$, the two methods are essentially identical in
shortfall.  As $\rho$ decreases, our shortfall drops while
the~\citet{FangSa19} shortfall is much less sensitive to correlation.
In the moderate range $\rho \in [0.3, 0.5]$, which is typical of
policy-evaluation scores for distinct treatment rules applied to the
same population, our method produces bounds that are $5$--$10\%$
tighter at equal or higher coverage.

The diversification advantage has a decision-theoretic interpretation.
In policy evaluation, distinct policies targeting different subgroups
can have similar expected values but weakly correlated scores.  The
simplex lower bound exploits this structure: the optimizer implicitly
hedges across the near-optimal policies, producing a tighter
confidence statement about the best attainable value.  This hedging is
a direct consequence of the simplex representation in
Proposition~\ref{prop:max-lower} and has no analogue in the
pointwise-max construction of~\citet{FangSa19}.


\section{Discussion and conclusion}
\label{sec:discussion}

The results of this paper rest on a single geometric observation: the
Euclidean profile statistic for a score-level empirical likelihood
equals the Mahalanobis distance from the score mean to a local level
set of the target map, and the first-order shape of that level set
determines both the limiting distribution and the correct bootstrap.
We close by situating this observation within the broader landscape of
semiparametric inference and by noting the boundaries of the present
theory.

The closest antecedents are the semiparametric EL results of
\citet{HjortMcKe09} and \citet{BravoEsVa20}, which establish Wilks
phenomena and bootstrap validity for estimating equations and smooth
functionals.  Our contribution is orthogonal: we profile the
\emph{target functional itself} rather than an estimating equation whose
solution equals the target.  For the best-policy-value problem, this
distinction is consequential---it lets us bypass both the projected
joint approach (which pays a $J$-dimensional critical value for a scalar
quantity) and the selected-policy approach (which ignores selection
uncertainty at ties).

The corrected multiplier bootstrap we derive for the nonregular case
builds on the same local-process information used by
\citet{FangSa19} and \citet{HongLi18,HongLi20} in their work on HDD
plug-in estimators.  The difference is in the object being calibrated.
Those papers produce corrected confidence intervals around a point
estimate; ours produces a profile-divergence confidence bound that can
be directly inverted through the simplex program of
Proposition~\ref{prop:max-lower}.  In structured optimization problems
where the tangent cone is explicit---best policy value, minimum risk
over a finite set, best subgroup treatment effect---this directness is a
practical advantage.

A unifying theme of the paper is that the uniqueness assumption
$|J_0|=1$ required by existing EL methods~\citep{DuchiGlNa21} is at
least as hard to verify as the inferential task itself: certifying a
strict optimum demands the same statistical evidence as the confidence
bound being constructed.  Our procedure avoids this circularity by
estimating the level-set geometry from the data, adapting automatically
to the hyperplane or cone case.  The simulation results in
Section~\ref{sec:sim} show that methods assuming uniqueness can undercover by several
percentage points when the assumption is violated.

The finite-policy setting we use as a motivating example gives especially
clean closed forms because the max functional has polyhedral level sets.
Extending the approach to continuous policy classes---where the level-set
geometry is richer---is a natural direction for future work.



\bibliographystyle{abbrvnat}
\setlength{\bibsep}{.7em}

\bibliography{main}

\newpage
\appendix

\section{Proofs}
\label{sec:proofs-geometry}

\subsection{Proof of Theorem~\ref{thm:exact-level-set}}
\label{proof:exact-level-set}

We work with the rescaled score-profile statistic. Noting that the constraint
\[
  \phi\!\left(\sum_{i=1}^n q_i X_{ni}\right)=\tau_0
\]
is equivalent to
\[
  \sum_{i=1}^n q_i X_{ni}=\theta_0+n^{-1/2}v
  \qquad\text{for some }v\in L_{1/\sqrt n},
\]
for $v \in \R^d$, define
\[
  \mathcal R_n(v)
  \defeq
  \inf\Biggl\{
  \sum_{i=1}^n (nq_i - 1)^2 :
  q \in \Delta_n,\
  \sum_{i=1}^n q_i X_{ni} = \theta_0 + n^{-1/2}v
  \Biggr\}
\]
so that $R_{n,2}(\tau_0)=\inf_{v\in L_{1/\sqrt n}}\mathcal R_n(v)$.
We rely on the following intermediate result whose proof we defer to
Section~\ref{proof:pointwise}, which shows that the Euclidean shift cost
reduces to a Mahalanobis distance.
\begin{proposition}[Pointwise shift cost]
  \label{prop:pointwise}
  
  Under Assumption~\ref{ass:score}, for every fixed $M < \infty$,
  \[
    \sup_{\norms{v} \le M}
    \bigl|
    \mathcal R_n(v) - (v - Z_n)^\top \what\Sigma_n^{-1}(v - Z_n)
    \bigr|
    \cp 0.
  \]
  With probability tending to one, equality holds for all
  $\norms{v} \le M$.
\end{proposition}

To apply the proposition, we must localize to a deterministic compact
set. Since $\norms{Z_n} = O_p(1)$ and $\what\Sigma_n \cp \Sigma \succ 0$, for
every $\varepsilon > 0$ there exist deterministic constants $K < \infty$ and
$0 < \underline\lambda < \overline\lambda < \infty$ such that
\begin{equation}
  \mathcal E_n \defeq \Biggl\{
  \norms{Z_n} \le K,\quad
  \underline\lambda I_d \preceq \what\Sigma_n \preceq \overline\lambda I_d
  \Biggr\} 
  \label{eq:good-event}
\end{equation}
has $\P(\mathcal{E}_n) \ge 1-\varepsilon$ for all sufficiently large $n$.
On $\mathcal E_n$, because $0 \in L_{1/\sqrt n}$,
\[
  R_{n,2}(\tau_0) \le \mathcal R_n(0).
\]
Apply Proposition~\ref{prop:pointwise} with $M_0 = K + 1$ (which satisfies
$\norms{Z_n} \le K < M_0$ on $\mathcal E_n$) to get
\[
  \mathcal R_n(0)
  = Z_n^\top \what\Sigma_n^{-1} Z_n + o_p(1)
  \le \underline\lambda^{-1} K^2 + o_p(1)
  ~~~\mbox{on the event}~~\mathcal{E}_n.
\]
On the same event, for any $v$ with
$\mathcal R_n(v) \le \mathcal R_n(0) + 1$,
\[
  \overline\lambda^{-1}\norms{v - Z_n}^2
  \le (v-Z_n)^\top \what\Sigma_n^{-1}(v-Z_n)
  \le \mathcal R_n(0) + 1
  \le \underline\lambda^{-1}K^2 + 2,
\]
so $\norms{v} \le K + \sqrt{\overline\lambda(\underline\lambda^{-1}K^2 + 2)}
\eqqcolon M$, a deterministic constant depending only on
$\varepsilon$.  Thus on $\mathcal E_n$,
\[
  R_{n,2}(\tau_0)
  = \inf_{v \in L_{1/\sqrt n} \cap B_M} \mathcal R_n(v).
\]

Now apply Proposition~\ref{prop:pointwise} with this deterministic~$M$:
\[
  \sup_{\norms{v}\le M}
  \left|
    \mathcal R_n(v)-(v-Z_n)^\top \what\Sigma_n^{-1}(v-Z_n)
  \right|\cp 0.
\]
Therefore on $\mathcal E_n$,
\[
  \inf_{v\in L_{1/\sqrt n}\cap B_M}\mathcal R_n(v)
  =
  \inf_{v\in L_{1/\sqrt n}\cap B_M}(v-Z_n)^\top \what\Sigma_n^{-1}(v-Z_n)+o_p(1)
  =
  d_{\what\Sigma_n}^2(Z_n,L_{1/\sqrt n})+o_p(1),
\]
where the last step uses the fact that minimizers of the Mahalanobis
distance lie in $B_M$ on $\mathcal E_n$.  Since $\P(\mathcal E_n) \ge 1 -
\varepsilon$ and $\varepsilon$ was arbitrary, the theorem follows.

\subsubsection{Proof of Proposition~\ref{prop:pointwise}}
\label{proof:pointwise}

Let
\[
  U_{ni}\defeq X_{ni}-\bar X_n,
  \qquad
  U_n=
  \begin{pmatrix}
    U_{n1}^\top\\
    \vdots\\
    U_{nn}^\top
  \end{pmatrix}\in\R^{n\times d}.
\]
Fix $v\in\R^d$ and write $a_i=nq_i-1$, so that $q_i=(1+a_i)/n$ and
$\sum_i a_i=0$. The constraint
\[
  \sum_{i=1}^n q_i X_{ni}=\theta_0+n^{-1/2}v
\]
is equivalent to
\[
  \bar X_n+\frac1n\sum_{i=1}^n a_i U_{ni}=\theta_0+n^{-1/2}v,
\]
or
\[
  \sum_{i=1}^n a_i U_{ni}=\sqrt n\,(v-Z_n).
\]
Ignoring temporarily the inequality constraints $1+a_i\ge 0$, the
problem becomes
\[
  \min_{a\in\R^n}\norms{a}_2^2
  \quad\text{subject to}\quad
  U_n^\top a = \sqrt n\,(v-Z_n).
\]
The minimum-norm solution is
\[
  a_n(v)=U_n(U_n^\top U_n)^{-1}\sqrt n\,(v-Z_n).
\]
Since
\[
  U_n^\top U_n = n\what\Sigma_n,
\]
we have
\[
  a_{ni}(v)=\frac1{\sqrt n}U_{ni}^\top \what\Sigma_n^{-1}(v-Z_n)
\]
and the minimum value is
\[
  \norms{a_n(v)}_2^2
  =
  (v-Z_n)^\top \what\Sigma_n^{-1}(v-Z_n).
\]

Fix $M<\infty$. By Assumption~\ref{ass:score}\ref{item:cov},
$\what\Sigma_n\cp\Sigma\succ 0$, so $\what\Sigma_n$ is invertible with
probability tending to one.  On that event,
\[
  \sup_{\norms{v}\le M}\max_{1\le i\le n}|a_{ni}(v)|
  \le
  \norms{\what\Sigma_n^{-1}}
  \left(M+\norms{Z_n}\right)
  \max_{1\le i\le n}\frac{\norms{U_{ni}}}{\sqrt n}.
\]
By Assumption~\ref{ass:score}, $\norms{Z_n}=O_p(1)$,
$\norms{\what\Sigma_n^{-1}}=O_p(1)$, and
\[
  \max_{1\le i\le n}\frac{\norms{U_{ni}}}{\sqrt n}\cp 0.
\]
Therefore
\[
  \sup_{\norms{v}\le M}\max_{1\le i\le n}|a_{ni}(v)|\cp 0.
\]
In particular, with probability tending to one,
\[
  1+a_{ni}(v)\ge 0
  \qquad\text{for all }\norms{v}\le M,\ 1\le i\le n.
\]
On that event the positivity constraints are inactive, so the
unconstrained minimum is the actual minimum and
\[
  \mathcal R_n(v)=(v-Z_n)^\top \what\Sigma_n^{-1}(v-Z_n)
  \qquad\text{for all }\norms{v}\le M.
\]

\subsection{Proof of Theorem~\ref{thm:general-set-limit}}
\label{proof:general-set-limit}

By Theorem~\ref{thm:exact-level-set},
\[
  R_{n,2}(\tau_0)=d_{\what\Sigma_n}^2(Z_n,L_{1/\sqrt n})+o_p(1).
\]
As in the proof of Theorem~\ref{thm:exact-level-set}, near-minimizers
are contained in a deterministic ball $B_M$ with arbitrarily high
probability. Fix such an $M$ and also fix $K<\infty$. On the event
$\norms{Z_n}\le K$ and $\norms{\what\Sigma_n^{-1}}\le K$, the function
\[
  q_n(v)\defeq(v-Z_n)^\top \what\Sigma_n^{-1}(v-Z_n)
\]
is Lipschitz on $B_M$ with constant at most $2K(M+K)$. Indeed, for
$u,v\in B_M$,
\[
  |q_n(u)-q_n(v)|
  \le
  \norms{\what\Sigma_n^{-1}}\,\norms{u-v}\,\norms{u+v-2Z_n}
  \le
  2K(M+K)\norms{u-v},
\]
where the last step uses
$\norms{u+v-2Z_n}\le \norms{u}+\norms{v}+2\norms{Z_n}\le 2(M+K)$.

Let
\[
  r_n(M)=d_H(L_{1/\sqrt n}\cap B_M,\ C\cap B_M).
\]
By Assumption~\ref{ass:levelset}, $r_n(M)\to 0$. For every
$u\in L_{1/\sqrt n}\cap B_M$ there exists $v\in C\cap B_M$ with
$\norms{u-v}\le r_n(M)$, so
\[
  \inf_{u\in L_{1/\sqrt n}\cap B_M} q_n(u)
  \ge
  \inf_{v\in C\cap B_M} q_n(v)-2K(M+K) r_n(M).
\]
The same argument with the roles of $L_{1/\sqrt n}$ and $C$ reversed
gives the opposite inequality. Therefore
\[
  \left|
    d_{\what\Sigma_n}^2(Z_n,L_{1/\sqrt n}\cap B_M)-d_{\what\Sigma_n}^2(Z_n,C\cap B_M)
  \right|
  \le 2K(M+K) r_n(M)
\]
on the event $\norms{Z_n}\le K$, $\norms{\what\Sigma_n^{-1}}\le K$.
Since $K$ is arbitrary and $r_n(M)\to 0$, we conclude that
\[
  d_{\what\Sigma_n}^2(Z_n,L_{1/\sqrt n})=
  d_{\what\Sigma_n}^2(Z_n,C)+o_p(1).
\]
Hence
\[
  R_{n,2}(\tau_0)=d_{\what\Sigma_n}^2(Z_n,C)+o_p(1).
\]

Finally, because $C$ is closed and $\Sigma$ is positive definite, the
map
\[
  (z,A)\mapsto d_A^2(z,C)
\]
is continuous at every $(z,A)$ with $A\succ 0$.
Assumption~\ref{ass:score} therefore implies
\[
  d_{\what\Sigma_n}^2(Z_n,C)\Rightarrow d_\Sigma^2(Z,C).
\]
This proves the theorem.

\subsection{Proof of Corollary~\ref{cor:hd}}
\label{proof:hd}

Let $H=\{v:a_0^\top v=0\}$. For any positive definite matrix $A$, the
$A^{-1}$-projection of $z$ onto $H$ is
\[
  z-Aa_0(a_0^\top A a_0)^{-1}a_0^\top z.
\]
Therefore
\[
  d_A^2(z,H)=\frac{(a_0^\top z)^2}{a_0^\top A a_0}.
\]
Applying this with $A=\what\Sigma_n$ and using
Theorem~\ref{thm:general-set-limit} yields
\[
  R_{n,2}(\tau_0)=\frac{(a_0^\top Z_n)^2}{a_0^\top \what\Sigma_n a_0}+o_p(1).
\]
Since $a_0^\top Z_n\Rightarrow \normal(0,a_0^\top\Sigma a_0)$ and
$a_0^\top\what\Sigma_n a_0\cp a_0^\top\Sigma a_0$, the ratio converges
to $\chi_1^2$.

\subsection{Proof of Theorem~\ref{thm:hd-boot}}
\label{proof:hd-boot}

Let
\[
  T(z,A,H)=d_A^2(z,H).
\]
For fixed hyperplane $H$ the map $(z,A)\mapsto T(z,A,H)$ is continuous
whenever $A\succ 0$. Because $H$ depends continuously on its normal
vector, the map $(z,A,\tilde H)\mapsto T(z,A,\tilde H)$ is continuous at
every $(z,A,H)$ with $A\succ0$.

By Assumption~\ref{ass:hdboot},
$\mathcal L^*(Z_n^* \mid X_{n1},\dots,X_{nn}) \Rightarrow_p \mathcal L(Z)$
and $\what H_n\cp H$.
Together with $\what\Sigma_n\cp \Sigma$, the conditional continuous
mapping theorem~\citep{VanDerVaartWe96} gives
\[
  \mathcal L^*\bigl(T_n^* \mid X_{n1},\dots,X_{nn}\bigr)
  \Rightarrow_p
  \mathcal L\bigl(d_\Sigma^2(Z,H)\bigr).
\]
By Corollary~\ref{cor:hd}, $d_\Sigma^2(Z,H)\sim \chi_1^2$, so
\[
  \mathcal L^*(T_n^*\mid X_{n1},\dots,X_{nn})\Rightarrow_p \mathcal L(\chi_1^2).
\]
If $R_{n,2}^*(\what\tau_n)=T_n^*+o_{\P^*}(1)$, Slutsky's theorem
yields the same conclusion for $R_{n,2}^*(\what\tau_n)$.

\subsection{Proof of Theorem~\ref{thm:hdd-boot}}
\label{proof:hdd-boot}

Let
\[
  T(z,A,S)=d_A^2(z,S).
\]
For any fixed closed set $S$, the map $(z,A)\mapsto T(z,A,S)$ is
continuous at every $(z,A)$ with $A\succ 0$. By
Assumption~\ref{ass:mult}, together with $\what\Sigma_n\cp\Sigma$,
\[
  \mathcal L_\xi\bigl(d_{\what\Sigma_n}^2(Z_n^\xi,C) \mid X_{n1},\dots,X_{nn}\bigr)
  \Rightarrow_p
  \mathcal L\bigl(d_\Sigma^2(Z,C)\bigr).
\]

It remains to replace $C$ by $\what C_n$. Fix $K<\infty$ and consider the
event $\norms{\what\Sigma_n^{-1}}\le K$ and
$\underline\lambda I_d\preceq \what\Sigma_n\preceq \overline\lambda I_d$.
We first show that the Mahalanobis-metric projections onto $C$ and
$\what C_n$ can be restricted to a deterministic ball. Because both $C$
and $\what C_n$ are closed cones containing the origin, $0$ is feasible,
so for any $z$ with $\norms{z}\le K$,
\[
  d_{\what\Sigma_n}^2(z,C)\le z^\top\what\Sigma_n^{-1}z\le K^3.
\]
If $p$ is the $\what\Sigma_n^{-1}$-projection of $z$ onto $C$, then
\[
  \overline\lambda^{-1}\norms{p-z}^2
  \le (p-z)^\top\what\Sigma_n^{-1}(p-z)\le K^3,
\]
giving $\norms{p}\le K+\sqrt{\overline\lambda K^3}$. The same bound
applies to the projection onto $\what C_n$. Set
$M\defeq K+\sqrt{\overline\lambda K^3}+1$. Then for $\norms{z}\le K$,
\[
  d_{\what\Sigma_n}^2(z,C)=d_{\what\Sigma_n}^2(z,C\cap B_M),
  \qquad
  d_{\what\Sigma_n}^2(z,\what C_n)=d_{\what\Sigma_n}^2(z,\what C_n\cap B_M).
\]

Now, on the event $\norms{z}\le K$ and $\norms{\what\Sigma_n^{-1}}\le K$,
the function
\[
  q_{n,\xi}(v)=(v-z)^\top\what\Sigma_n^{-1}(v-z)
\]
is Lipschitz on $B_M$ with constant at most $2K(M+K)$, since
$\norms{u+v-2z}\le 2(M+K)$ for $u,v\in B_M$. Therefore
\[
  \sup_{\norms{z}\le K}
  \left|
    d_{\what\Sigma_n}^2(z,\what C_n\cap B_M)-d_{\what\Sigma_n}^2(z,C\cap B_M)
  \right|
  \le
  2K(M+K)\,d_H(\what C_n\cap B_M,C\cap B_M).
\]
By Assumption~\ref{ass:mult}, the right-hand side is $o_p(1)$.

Because $Z_n^\xi$ is conditionally tight, for every $\varepsilon>0$
there exists $K<\infty$ such that
\[
  \P\left\{
    \P_\xi(\norms{Z_n^\xi}>K\mid X_{n1},\dots,X_{nn})>\varepsilon
  \right\}<\varepsilon
\]
for all sufficiently large $n$. On the complement of that event,
the localization above applies with $z=Z_n^\xi$, giving
\[
  d_{\what\Sigma_n}^2(Z_n^\xi,\what C_n)-d_{\what\Sigma_n}^2(Z_n^\xi,C)=o_{\P_\xi}(1).
\]
Thus
\[
  T_n^\xi=d_{\what\Sigma_n}^2(Z_n^\xi,\what C_n)
  =
  d_{\what\Sigma_n}^2(Z_n^\xi,C)+o_{\P_\xi}(1),
\]
so $\mathcal L_\xi(T_n^\xi \mid X_{n1},\dots,X_{nn})
\Rightarrow_p \mathcal L(d_\Sigma^2(Z,C))$.
The Kolmogorov-distance statement follows when the limit
cdf is continuous.



\subsection{Dual representation}
\label{proof:dual-representation}

\begin{proposition}[Dual representation]
  \label{prop:dual-representation}
  Under Assumption~\ref{ass:score}, for every $v\in\R^d$,
  \[
    \mathcal R_{n,f}(v)
    =
    \sup_{\nu\in\R,\,\lambda\in\R^d}
    \left\{
      \nu+\lambda^\top(\theta_0+n^{-1/2}v)
      -
      \frac1n\sum_{i=1}^n f^*(\nu+\lambda^\top X_{ni})
    \right\},
  \]
  where $f^*$ is the Fenchel conjugate of~$f$.  In centered form,
  \[
    \mathcal R_{n,f}(v)
    =
    \sup_{\nu,\lambda}
    \left\{
      n^{-1/2}\lambda^\top(v-Z_n)
      -
      \frac1n\sum_{i=1}^n
      \bigl[f^*(\nu+\lambda^\top U_{ni})-\nu-\lambda^\top U_{ni}\bigr]
    \right\},
  \]
  with $U_{ni}=X_{ni}-\bar X_n$.
\end{proposition}

\begin{proof}
Write $m_n(v)=\theta_0+n^{-1/2}v$. Using the variables $t_i=nq_i$, the
primal problem becomes
\[
  \mathcal R_{n,f}(v)
  =
  \inf_{t_i\ge 0}
  \left\{
    \frac1n\sum_{i=1}^n f(t_i):
    \frac1n\sum_{i=1}^n t_i=1,\
    \frac1n\sum_{i=1}^n t_i X_{ni}=m_n(v)
  \right\}.
\]
Its Lagrangian is
\[
  \mathcal L(t,\nu,\lambda)
  =
  \frac1n\sum_{i=1}^n f(t_i)
  +
  \nu\left(1-\frac1n\sum_{i=1}^n t_i\right)
  +
  \lambda^\top\left(m_n(v)-\frac1n\sum_{i=1}^n t_i X_{ni}\right),
\]
with $\nu\in\R$ and $\lambda\in\R^d$. Therefore
\[
  \inf_{t_i\ge 0}\mathcal L(t,\nu,\lambda)
  =
  \nu+\lambda^\top m_n(v)
  -
  \frac1n\sum_{i=1}^n
  \sup_{t_i\ge 0}
  \left\{
    t_i(\nu+\lambda^\top X_{ni})-f(t_i)
  \right\}.
\]
By definition of the Fenchel conjugate,
\[
  \sup_{t_i\ge 0}
  \left\{
    t_i(\nu+\lambda^\top X_{ni})-f(t_i)
  \right\}
  =
  f^*(\nu+\lambda^\top X_{ni}).
\]
Since the primal problem is convex with affine constraints and the
feasible point $t_i\equiv 1$ is interior, strong duality holds. Hence
\[
  \mathcal R_{n,f}(v)
  =
  \sup_{\nu,\lambda}
  \left\{
    \nu+\lambda^\top m_n(v)
    -
    \frac1n\sum_{i=1}^n f^*(\nu+\lambda^\top X_{ni})
  \right\}.
\]
This is the first display in
Proposition~\ref{prop:dual-representation}.

For the centered form, write $X_{ni}=\bar X_n+U_{ni}$ and note that
\[
  \bar X_n=\theta_0+n^{-1/2}Z_n,
  \qquad
  \frac1n\sum_{i=1}^n U_{ni}=0.
\]
Then
\begin{align*}
  \nu+\lambda^\top m_n(v)
  -\frac1n\sum_{i=1}^n f^*(\nu+\lambda^\top X_{ni})
  &=
  \nu+\lambda^\top(\theta_0+n^{-1/2}v)
  -\frac1n\sum_{i=1}^n f^*(\nu+\lambda^\top\bar X_n+\lambda^\top U_{ni})\\
  &=
  n^{-1/2}\lambda^\top(v-Z_n)
  -\frac1n\sum_{i=1}^n
  \Bigl[
  f^*(\nu+\lambda^\top U_{ni})-\nu-\lambda^\top U_{ni}
  \Bigr],
\end{align*}
after reparameterizing $\nu\leftarrow \nu+\lambda^\top\bar X_n$ and
using the centering identity
$\frac1n\sum_i \lambda^\top U_{ni}=0$. This gives the second display.
\end{proof}


\section{Proofs of policy-specific results}
\label{sec:proofs-policy}

\subsection{Proof of Proposition~\ref{prop:max-levelset}}
\label{proof:max-levelset}

Fix $M<\infty$ and $0<t<\Delta_0/M$. Let $v\in B_M$.

Suppose first that $\max_{j\in J_0}v_j=0$. Then for every $j\in J_0$,
\[
  \theta_{0j}+t v_j\le \tau_0,
\]
and equality holds for at least one $j\in J_0$. For every $k\notin J_0$,
\[
  \theta_{0k}+t v_k
  \le
  \theta_{0k}+tM
  <
  \theta_{0k}+\Delta_0
  \le
  \tau_0.
\]
Hence
\[
  \max_{1\le j\le J}(\theta_{0j}+tv_j)=\tau_0,
\]
so $v\in L_t$.

Conversely, suppose $v\in L_t\cap B_M$. Then
\[
  \max_{1\le j\le J}(\theta_{0j}+tv_j)=\tau_0.
\]
For every $k\notin J_0$,
\[
  \theta_{0k}+tv_k\le \theta_{0k}+tM<\tau_0,
\]
so the maximum must be attained inside $J_0$. Therefore
\[
  \max_{j\in J_0}(\theta_{0j}+tv_j)=\tau_0.
\]
But $\theta_{0j}=\tau_0$ for all $j\in J_0$, so this is equivalent to
\[
  \max_{j\in J_0} v_j = 0.
\]
This proves the first display. The rest of the proposition is immediate.

\subsection{Proof of Proposition~\ref{prop:max-face}}
\label{proof:max-face}

Fix a nonempty subset $A\subseteq\{1,\dots,J\}$ and write $r=|A|$. The
affine face
\[
  F_A(\tau)=\{m\in\R^J:m_j=\tau\ \text{for every }j\in A\}
\]
can be written as
\[
  C_A m = \tau \onevec.
\]
The metric projection of $\bar X_n$ onto $F_A(\tau)$ under the quadratic
form generated by $\what\Sigma_n^{-1}$ is the solution to
\[
  \min_m (m-\bar X_n)^\top\what\Sigma_n^{-1}(m-\bar X_n)
  \quad\text{subject to}\quad
  C_A m = \tau\onevec.
\]
The Lagrangian first-order conditions give
\[
  m=\bar X_n-\what\Sigma_n C_A^\top\lambda
\]
and
\[
  C_A\bar X_n-C_A\what\Sigma_n C_A^\top\lambda=\tau\onevec.
\]
Therefore
\[
  \lambda=(C_A\what\Sigma_n C_A^\top)^{-1}(C_A\bar X_n-\tau\onevec)
\]
and the projection point is
\[
  \Pi_A(\tau)=\bar X_n-\what\Sigma_n C_A^\top (C_A\what\Sigma_n C_A^\top)^{-1}(C_A\bar X_n-\tau\onevec).
\]
Substituting this into the objective yields
\[
  (\Pi_A(\tau)-\bar X_n)^\top \what\Sigma_n^{-1}(\Pi_A(\tau)-\bar X_n)
  =
  (C_A\bar X_n-\tau\onevec)^\top
  (C_A\what\Sigma_n C_A^\top)^{-1}
  (C_A\bar X_n-\tau\onevec).
\]
Multiplying by $n$ gives the displayed formula for
$R_{n,A}^{\max}(\tau)$.

The projection onto $M_\tau=\{m:\max_j m_j=\tau\}$ must lie on one of
the faces $F_A(\tau)$ and is feasible exactly when the nonbinding
coordinates do not exceed $\tau$. Minimizing over all feasible faces
gives the second display.

\subsection{Proof of Proposition~\ref{prop:max-lower}}
\label{proof:max-lower}

Because
\[
  \max_{1\le j\le J}m_j=\max_{w\in\Delta_J} w^\top m,
\]
we have
\[
  L_n(c)=\inf_{m\in E_n(c)}\max_{w\in\Delta_J} w^\top m.
\]
The set $E_n(c)$ is compact and convex, $\Delta_J$ is compact and
convex, and the map $(m,w)\mapsto w^\top m$ is bilinear. Sion's minimax
theorem therefore implies
\[
  L_n(c)=\max_{w\in\Delta_J}\inf_{m\in E_n(c)} w^\top m.
\]
For fixed $w$, the inner problem is the support function of an ellipsoid
in the negative direction:
\[
  \inf_{m\in E_n(c)} w^\top m
  =
  w^\top \bar X_n - \sqrt{\frac{c}{n}}\sqrt{w^\top \what\Sigma_n w}.
\]
Substituting this expression yields the proposition.

\subsection{Proof of Proposition~\ref{prop:unique-inflation}}
\label{proof:unique-inflation}

By Proposition~\ref{prop:max-levelset}, if $J_0=\{j_0\}$ and
$\Delta_0>0$, then for every fixed $M<\infty$ and every sufficiently
small $t$,
\[
  L_t\cap B_M=\{v\in B_M:v_{j_0}=0\}.
\]
Hence, for $\tau$ within $O(n^{-1/2})$ of $\tau_0$, the relevant face of
the max-level set is asymptotically the single coordinate face
$m_{j_0}=\tau$. Applying Proposition~\ref{prop:max-face} with
$A=\{j_0\}$ yields
\[
  R_{n,2}^{\max}(\tau)
  =
  n\frac{(\bar X_{n,j_0}-\tau)^2}{\what\Sigma_{n,j_0j_0}}
  +
  o_p(1)
\]
uniformly for $|\tau-\tau_0|\le M n^{-1/2}$.

The direct profile lower bound is defined by
$R_{n,2}^{\max}(\tau)\le c_{1,\alpha}$. Solving the preceding quadratic
approximation for $\tau$ gives
\[
  L_{n,\mathrm{dir}}(1-\alpha)
  =
  \bar X_{n,j_0}
  -
  \sqrt{\frac{c_{1,\alpha}}{n}}\sqrt{\what\Sigma_{n,j_0j_0}}
  +
  o_p(n^{-1/2}).
\]

For the projected joint bound, the feasible set is the ellipsoid
\[
  E_n(c_{J,\alpha})
  =
  \left\{
    m:
    n(m-\bar X_n)^\top\what\Sigma_n^{-1}(m-\bar X_n)\le c_{J,\alpha}
  \right\}.
\]
Because the optimizer is unique and separated, the minimizer of
$\max_j m_j$ over $E_n(c_{J,\alpha})$ is attained on the same
single-coordinate face to first order. Equivalently,
Proposition~\ref{prop:max-lower} with $c=c_{J,\alpha}$ has optimizer
$w=e_{j_0}+o_p(1)$. Therefore
\[
  L_{n,\mathrm{joint}}(1-\alpha)
  =
  \bar X_{n,j_0}
  -
  \sqrt{\frac{c_{J,\alpha}}{n}}\sqrt{\what\Sigma_{n,j_0j_0}}
  +
  o_p(n^{-1/2}).
\]
Taking the ratio of the two radii yields the last display.

\subsection{Proof of Proposition~\ref{prop:active-set}}
\label{proof:active-set}

Because $\what\Sigma_{n,jj} \cp \Sigma_{jj} > 0$ by
Assumption~\ref{ass:score}\ref{item:cov} and $J$ is fixed,
$\max_j \sqrt{\what\Sigma_{n,jj}} \cp \max_j \sqrt{\Sigma_{jj}}
\eqqcolon \sigma_{\max} > 0$.  Therefore
$\kappa_n = (\sigma_{\max} + o_p(1)) \sqrt{\log n / n}$, which
gives $\kappa_n \cp 0$ (since $\log n / n \to 0$) and
$\sqrt{n}\,\kappa_n = (\sigma_{\max} + o_p(1)) \sqrt{\log n}
\to \infty$.

For (i), if $j \in J_0$ then $\theta_{0j} = \tau_0$ and
$\bar X_{n,j} = \tau_0 + O_p(n^{-1/2})$, while
$\max_k \bar X_{n,k} = \tau_0 + O_p(n^{-1/2})$.  So
$\max_k \bar X_{n,k} - \bar X_{n,j} = O_p(n^{-1/2})$, which is
$o_p(\kappa_n)$ because $\sqrt{n}\,\kappa_n \to \infty$.
Hence $j \in \what J_n$ with probability tending to one.  Conversely,
if $k \notin J_0$ then
$\min_{k \notin J_0}(\tau_0 - \theta_{0k}) = \Delta_0 > 0$ (since
$J$ is finite and all gaps are strictly positive), so
$\max_\ell \bar X_{n,\ell} - \bar X_{n,k}
\ge \Delta_0 + o_p(1)$,
which exceeds $\kappa_n$ eventually, giving $k \notin \what J_n$.
When $J_0 = \{1,\dots,J\}$ (all policies optimal), there are no
suboptimal policies to exclude, and $\what J_n \supseteq J_0$ implies
$\what J_n = J_0$ directly.

For (ii), on the event $\what J_n = J_0$ the cones coincide exactly:
$\what C_n = C$.  Part~(i) gives
$\P(\what J_n = J_0) \to 1$, so $d_H(\what C_n \cap B_M, C \cap B_M)
= 0$ with probability tending to one.


\end{document}